\documentclass{article}

\usepackage{graphicx}
\usepackage{amsmath}
\usepackage{amsfonts}
\usepackage{amssymb}
\usepackage{amsthm}
\usepackage{multirow}
\usepackage{multicol}
\usepackage{bm}
\usepackage{enumitem}

\newtheorem{theorem}{Theorem}[section]
\newtheorem{lemma}[theorem]{Lemma}
\newtheorem{conjecture}[theorem]{Conjecture}

\begin{document}

\title{More Tight Bounds for Active Self-Assembly Using an Insertion Primitive\thanks{Some of the results in this work have been published as~\cite{Malchik-2014a}.}}

\author{Benjamin Hescott\footnote{Tufts University, Department of Computer Science, \texttt{hescott@cs.tufts.edu}} 
\and Caleb Malchik\footnote{Tufts University, Department of Computer Science, \texttt{caleb.malchik@tufts.edu}}  
\and Andrew Winslow\footnote{Universit\'e Libre de Bruxelles, D\'{e}partement d'Informatique, \texttt{awinslow@ulb.ac.be}}
}

\date{}

\maketitle

\begin{abstract}
We prove several limits on the behavior of a model of self-assembling particles introduced by Dabby and Chen (SODA 2013), called \emph{insertion systems}, where monomers insert themselves into the middle of a growing linear polymer.
First, we prove that the expressive power of these systems is equal to context-free grammars, answering a question posed by Dabby and Chen.

Second, we give tight bounds on the maximum length and minimum expected time of constructed polymers in systems of three increasingly restricted classes.
We prove that systems of $k$ monomer types can deterministically construct polymers of length $n = 2^{\Theta(k^{3/2})}$ in $O(\log^{5/3}(n))$ expected time.
We also prove that if non-deterministic construction of a finite number of polymers is permitted, then the expected construction time can be reduced to $O(\log^{3/2}(n))$ at the trade-off of decreasing the length to $2^{\Theta(k)}$.
If the system is allowed to construct an infinite number of polymers, then constructing polymers of unbounded length in $O(\log{n})$ expected time is possible.
We follow these positive results with a set of lower bounds proving that these are the best possible polymer lengths and expected construction times.
\end{abstract}

\section{Introduction}

In this work we study a theoretical model of \emph{algorithmic self-assembly}, in which simple particles aggregate in a distributed manner to carry out complex functionality.
Perhaps the the most well-studied theoretical model of algorithmic self-assembly is the \emph{abstract Tile Assembly Model (aTAM)} of Winfree~\cite{Winfree-1998a} consisting of square \emph{tiles} irreversibly attach to a growing polyomino-shaped assembly according to matching edge colors.
This model is capable of Turing-universal computation~\cite{Winfree-1998a}, self-simulation~\cite{Doty-2012b}, and efficient assembly of general (scaled) shapes~\cite{Soloveichik-2007a} and squares~\cite{Adleman-2001a,Rothemund-2000a}.
Despite this power, the model is incapable of assembling some shapes efficiently; a single row of $n$ tiles requires $n$ distinct tile types and $\Omega(n)$ expected assembly time~\cite{Adleman-2001a} and any shape with $n$ tiles requires $\Omega(\sqrt{n})$ expected time to assemble~\cite{Keenan-2013a}.

Such a limitation may not seem so significant, except that a wide range of biological systems form complex assemblies in time polylogarithmic in the assembly size, as Dabby and Chen~\cite{Dabby-2013a} and Woods et al.~\cite{Woods-2013b} observe.
These biological systems are capable of such growth because their particles (e.g.\ living cells) \emph{actively} carry out geometric reconfiguration.
In the interest of both understanding naturally occurring biological systems and creating synthetic systems with additional capabilities, several models of \emph{active self-assembly} have been proposed recently.
These include the graph grammars of Klavins et al.~\cite{Klavins-2004b,Klavins-2004a}, the \emph{nubots} model of Woods et al.~\cite{Chen-2014a,Chen-2013a,Woods-2013b}, and the insertion systems of Dabby and Chen~\cite{Dabby-2013a}. 
Both graph grammars and nubots are capable of a topologically rich set of assemblies and reconfigurations, but rely on stateful particles forming complex bond arrangements.
In contrast, insertion systems consist of stateless particles forming a single chain of bonds.
Indeed, all insertion systems are captured as a special case of nubots in which a linear polymer is assembled via parallel insertion-like reconfigurations, as in Theorem 5.1 of~\cite{Woods-2013a}. 
The simplicity of insertion systems makes their implementation in matter a more immediately attainable goal; Dabby and Chen~\cite{Dabby-2013b,Dabby-2013a} describe a direct implementation of these systems in DNA. 

We are careful to make a distinction between \emph{active self-assembly} where assemblies undergo reconfiguration, and \emph{active tile self-assembly}~\cite{Gautam-2013a,Hendricks-2013a,Jonoska-2014a,Jonoska-2014b,Keenan-2013b,Majumder-2008a,Padilla-2012a,Padilla-2014a}, where tile-based assemblies change their bond structure.
Active self-assembly enables exponential assembly rates by enabling insertion of new particles throughout the assembly, while active tile self-assembly does not: assemblies formed consist of rigid tiles and the $\Omega(\sqrt{n})$ expected-time lower bound of Keenan, Schweller, Sherman, and Zhong~\cite{Keenan-2013a} still applies.

\subsection{Our results}

We prove two types of results on the behavior of insertion systems.
We start by considering what languages can be \emph{expressed} by insertion systems, i.e. correspond to a set of polymers constructed by some insertion system.
Dabby and Chen prove that only context-free languages are expressible by insertion systems, and ask whether every context-free language is indeed expressed by some insertion system. 
We answer this question in the affirmative, and as a consequence prove that the languages expressible by insertion systems are exactly the context-free languages.

After achieving a tight bound on the expressive power of insertion systems, we turn to considering the efficiency of insertion systems, both with regards to the number of monomer types used and the expected time to construct polymers.
Dabby and Chen prove that insertion systems with $k$ monomer types can deterministically construct polymers of length $n = 2^{\Theta(\sqrt{k})}$ in $O(\log^3{n})$ expected time.
In Section~\ref{sec:positive-results} we describe three constructions in this vein.
First, we improve on the result of Dabby and Chen, proving that deterministic construction of polymers with length $n = 2^{\Theta(k^{3/2})}$ in $O(\log^{5/3}(n))$ expected time is possible (Theorem~\ref{thm:types-extreme-ub}).
Second, we prove that allowing non-deterministic construction of a finite set of polymers enables constructing polymers of length $n = 2^{\Theta(k)}$ in $O(\log^{3/2}(n))$ expected time (Theorem~\ref{thm:speed-extreme-ub}).
Third, we briefly describe a 2-monomer-type system constructing polymers of all lengths $n \geq 3$ in $O(\log{n})$ expected time, which can easily be seen to be optimal for unrestricted insertion systems.

In Section~\ref{sec:negative-results}, we that prove these systems are each optimal with regards to both polymer length and expected construction time.
First, we prove that deterministically constructing a polymer of length $n$ takes $\Omega(\log^{5/3}(n))$ expected time (Theorem~\ref{thm:deterministic-lb}), matching the construction time of Theorem~\ref{thm:types-extreme-ub} and proving that no trade-off between monomer types and construction is possible for deterministic systems.

Next, we prove that constructing a polymer of length $n$ in a system constructing a finite set of polymers, including deterministic systems, requires $\Omega(\log^{2/3}(n))$ monomer types and, if $\Theta(\log^{2/3}(n))$ types are used, $\Omega(\log^{5/3}(n))$ expected time (Theorem~\ref{thm:types-extreme-lb}).
Both of these bounds match those achieved by the construction of Theorem~\ref{thm:types-extreme-ub}.

Finally, we prove that constructing a polymer of length $n$ in a system constructing a finite set of polymers requires $\Omega(\log^{3/2}(n))$ expected time and, if $\Theta(\log^{3/2}(n))$ expected time is achieved, $\Omega(\log{n})$ monomer types (Theorem~\ref{thm:speed-extreme-lb}). 
Again, both of these bounds match those achieved by the construction of Theorem~\ref{thm:speed-extreme-ub}.
 
Taken together, these results give an asymptotically tight characterization of maximum length and minimum expected time of polymer constructions for three general classes of insertion systems: deterministic construction, construction of a finite set of polymers, and unrestricted construction.
Our lower bounds also imply a length and time tradeoff for systems constructing a finite set of polymers: constructing a polymer of length $n$ using $k$ monomer types takes $\Omega(\log^2(n)/\sqrt{k})$ (Lemma~\ref{lem:trade-off-lb}).

\section{Definitions}

Section~\ref{sec:grammar-defns} defines standard context-free grammars, as well as a special type called \emph{pair grammars}, used in Section~\ref{sec:expressive-power}.
Section~\ref{sec:is-defns} defines insertion systems, with a small number of modifications from the definitions given in~\cite{Dabby-2013a} designed to ease readability.
Section~\ref{sec:expressive-power-defn} formalizes the notion of expressive power used in~\cite{Dabby-2013a}.

\subsection{Grammars}
\label{sec:grammar-defns}

A \emph{context-free grammar} $\mathcal{G}$ is a 4-tuple $\mathcal{G} = (\Sigma, \Gamma, \Delta, S)$.
The sets $\Sigma$ and $\Gamma$ are the \emph{terminal} and \emph{non-terminal symbols} of the grammar.
The set $\Delta$ consists of \emph{production rules} or simply \emph{rules}, each of the form $L \rightarrow R_1 R_2 \cdots R_j$ with $L \in \Gamma$ and $R_i \in \Sigma \cup \Gamma$.
Finally, the symbol $S \in \Gamma$ is a special \emph{start symbol}.
The \emph{language of $\mathcal{G}$}, denoted $L(\mathcal{G})$, is the set of finite strings that can be \emph{derived} by starting with $S$, and repeatedly replacing a non-terminal symbol found on the left-hand side of some rule in $\Delta$ with the sequence of symbols on the right-hand side of the rule.
The \emph{size} of $\mathcal{G}$ is $|\Delta|$, the number of rules in $\mathcal{G}$.
If every rule in $\Delta$ is of the form $L \rightarrow R_1 R_2$ or $L \rightarrow t$, with $R_1 R_2 \in \Gamma$ and $t \in \Sigma$, then the grammar is said to be in \emph{Chomsky normal form}.

An \emph{integer-pair grammar}, used in Section~\ref{sec:expressive-power}, is a context-free grammar in Chomsky normal form such that each non-terminal symbol is an integer pair $(a, d)$, and each production rule has the form $(a, d) \rightarrow (a, b) (c, d)$ or $(a, d) \rightarrow t$.

\subsection{Insertion systems}
\label{sec:is-defns}

An \emph{insertion system} in the active self-assembly model of Dabby and Chen~\cite{Dabby-2013a} carries out the construction of a linear \emph{polymer} consisting of constant length \emph{monomers}.
A polymer grows incrementally by the insertion of a monomer at an \emph{insertion site} between two existing monomers in the polymer, according to complementary bonding sites between the monomer and the insertion site.

An insertion system $\mathcal{S}$ is defined as a 4-tuple $\mathcal{S} = (\Sigma, \Delta, Q, R)$.
The first element, $\Sigma$, is a set of symbols.
Each symbol $s \in \Sigma$ has a \emph{complement} $s^*$.
We denote the complement of a symbol $s$ as $\overline{s}$, i.e. $\overline{s} = s^*$ and $\overline{s^*} = s$.
The set $\Delta$ is a set of \emph{monomer types}, each assigned a \emph{concentration}.
Each monomer is specified by a quadruple $(a, b, c, d)^+$ or $(a, b, c, d)^-$, where $a, b, c, d \in \Sigma \cup \{s^* : s \in \Sigma\}$, and each concentration is a real number between~0 and~1.
The sum of all concentrations in $\Delta$ must be at most~1.
The two symbols $Q = (a, b)$ and $R = (c, d)$ are special two-symbol monomers that together form the \emph{initiator} of $\mathcal{S}$.
It is required that either $\overline{a} = d$ or $\overline{b} = c$.
The \emph{size} of $\mathcal{S}$ is $|\Delta|$, the number of monomer types in $\mathcal{S}$.

A \emph{polymer} is a sequence of monomers $Q m_1 m_2 \dots m_n R$ where $m_i \in \Delta$ such that for each pair of adjacent monomers $(w, x, a, b) (c, d, y, z)$, either $\overline{a} = d$ or $\overline{b} = c$.
The \emph{length} of a polymer is the number of monomers, including $Q$ and $R$, it contains.
Each pair of adjacent monomer ends $(a, b) (c, d)$ form an \emph{insertion site}.
Monomers can be inserted into an insertion site $(a, b) (c, d)$ according to the following rules (see Figure~\ref{fig:figure}):

\begin{enumerate}
\item If $\overline{a} = d$, then any monomer $(\overline{b}, e, f, \overline{c})^+$ can be inserted.
\item If $\overline{b} = c$, then any monomer $(e, \overline{a}, \overline{d}, f)^-$ can be inserted.\footnote{In~\cite{Dabby-2013a}, this rule is described as a monomer $(\overline{d}, f, e, \overline{a})^-$ that is inserted into the polymer as $(e, \overline{a}, \overline{d}, f)$.}
\end{enumerate}

\begin{figure}[ht]
\centering
\includegraphics[scale=1.0]{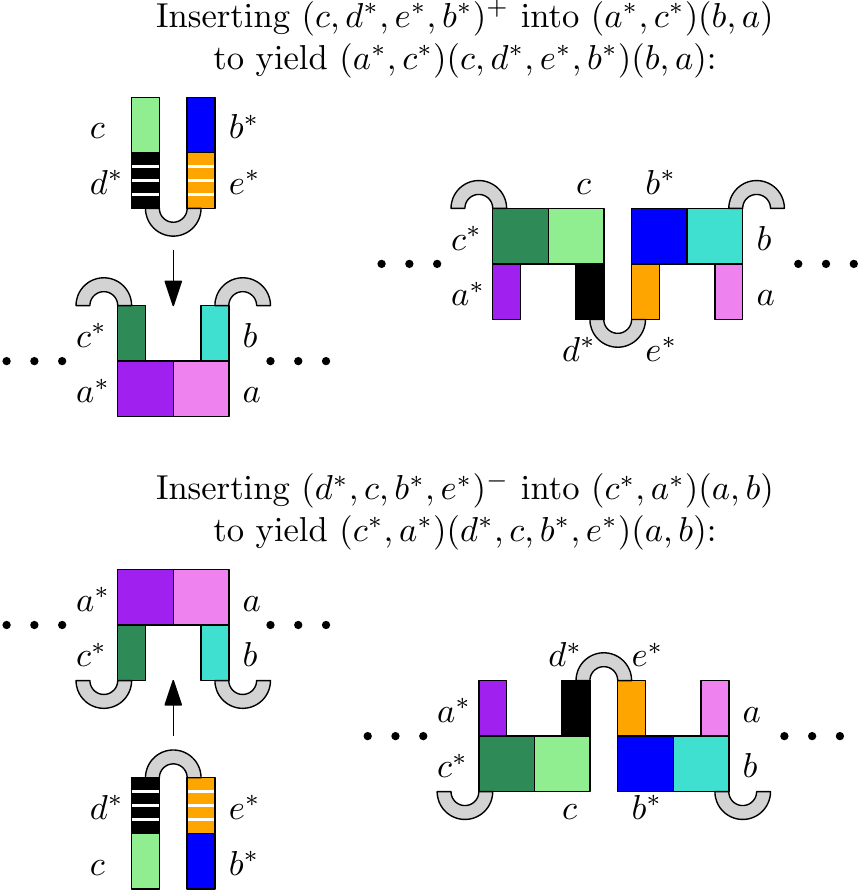}
\caption{A pictorial interpretation of the two insertion rules for monomers. Loosely based on Figure~2 and corresponding DNA-based implementation of~\cite{Dabby-2013a}.}
\label{fig:figure}
\end{figure}

A monomer is inserted after time $t$, where $t$ is an exponential random variable with rate equal to the concentration of the monomer type.
The set of all polymers \emph{constructed} by an insertion system is recursively defined as any polymer constructed by inserting a monomer into a polymer constructed by the system, beginning with the initiator.
Note that the insertion rules guarantee by induction that for every insertion site $(a, b) (c, d)$, either $\overline{a} = d$ or $\overline{b} = c$. 

We say that a polymer is \emph{terminal} if no monomer can be inserted into any insertion site in the polymer, and that an insertion system \emph{deterministically constructs} a polymer $P$ if every polymer constructed by the system is either $P$ or is non-terminal and has length less than that of $P$ (i.e. can become $P$).
The \emph{string representation} of a polymer is the sequence of symbols found on the polymer from left to right, e.g. $(a, b) (b^*, a, d, c) (c^*, a)$ has string representation $abb^*adcc^*a$.
We call the set of string representations of all terminal polymers of an insertion system $\mathcal{S}$ the \emph{language} of $\mathcal{S}$, denoted $L(\mathcal{S})$. 


\subsection{Expressive power}
\label{sec:expressive-power-defn}

Intuitively, a system \emph{expresses} another if the terminal polymers or strings created by the system ``look'' like the terminal polymers or strings created by the other system. 
In the simplest instance, an integer-pair grammar $\mathcal{G}'$ is said to \emph{express} a context-free grammar $\mathcal{G}$ if $L(\mathcal{G}') = L(\mathcal{G})$.
Similarly, a grammar $\mathcal{G}$ is said to \emph{express} an insertion system $\mathcal{S}$ if $L(\mathcal{S}) = L(\mathcal{G})$, i.e. if the set of string representations of the terminal polymers of $\mathcal{S}$ equals the language of $\mathcal{G}$. 

An insertion system $\mathcal{S} = (\Sigma', \Delta', Q', R')$ is said to express a grammar $\mathcal{G} = (\Sigma, \Gamma, \Delta, S)$ if there exists a function $g : \Sigma' \cup \{s^* : s \in \Sigma'\} \rightarrow \Sigma \cup \{\varepsilon\}$ such that $\{g(s_1') g(s_2') \dots g(s_n') : s_1' s_2' \dots s_n' \in L(\mathcal{S})\} = L(\mathcal{G})$.
More precisely, we require that there exists a fixed integer $\kappa$ such that for any substring $s_{i+1}' s_{i+2}' \dots s_{i+\kappa}'$ in a string in $L(\mathcal{S})$, $\{g(s_{i+1}'), g(s_{i+2}'), \dots, g(s_{i+\kappa}')\} \neq \{\varepsilon\}$. 
That is, the insertion system symbols mapping to grammar terminal symbols are evenly distributed throughout the polymer.
The requirement of a fixed integer $\kappa$ prevents the possibility of a polymer containing arbitrarily long and irregular regions of ``garbage'' monomers.

\section{The Expressive Power of Insertion Systems}
\label{sec:expressive-power}

Dabby and Chen proved that any insertion system has a context-free grammar expressing it.
They construct such a grammar by creating a non-terminal for every possible insertion site 
and a production rule for every monomer type insertable into the site.
For instance, the insertion site $(a,b)(c^*,a^*)$ and monomer type $(b^*, d^*, e, c)^+$ induce non-terminal symbol $A_{(a, b)(c^*, a^*)}$ and production rule $A_{(a, b)(c^*, a^*)} \rightarrow A_{(a,b)(b^*, d^*)} A_{(e, c)(c^*,a^*)}$.
Here we give a reduction in the other direction, resolving in the affirmative the question posed by Dabby and Chen of whether context-free grammars and insertion systems have the same expressive power:

\begin{theorem}
\label{thm:IS-express-CFG}
For every context-free grammar $G$, there exists an insertion system that expresses $G$.
\end{theorem}

The primary difficulty in proving Theorem~\ref{thm:IS-express-CFG} lies in developing a way to simulate the ``complete'' replacement that occurs during derivation with the ``incomplete'' replacement that occurs at an insertion site during insertion.
For instance, $bcAbc \Rightarrow bcDDbc$ via a production rule $A \rightarrow DD$ and $A$ is completely replaced by $DD$.
On the other hand, inserting a monomer $(b^*, d, d, c)^+$ into a site $(a, b) (c^*, a^*)$ yields the consecutive sites $(a, b) (b^*, d)$ and $(d, c) (c^*, a^*)$, with $(a, b) (c^*, a^*)$ only partially replaced -- the left side of the first site and the right side of second site together form the initial site.
This behavior constrains how replacement can be captured by insertion sites, and the $\kappa$ parameter of the definition of expression (Section~\ref{sec:expressive-power-defn}) prevents eliminating the issue via additional insertions.

We overcome this difficulty by proving Theorem~\ref{thm:IS-express-CFG} in two steps. 
First, we prove that integer-pair grammars, a constrained type of grammar with incomplete replacements, are able to express context-free grammars (Lemma~\ref{lem:PG-express-CFG}).
Second, we prove integer-pair grammars can be expressed by insertion systems (Lemma~\ref{lem:IS-express-PG}).

\begin{lemma}
\label{lem:PG-express-CFG}
For every context-free grammar $\mathcal{G}$, there exists an integer-pair grammar that expresses $\mathcal{G}$.
\end{lemma}

\begin{proof}
Let $\mathcal{G} = (\Sigma, \Gamma, \Delta, S)$. 
Let $n = |\Gamma|$.
Start by putting $\mathcal{G}$ into Chomsky normal form and then relabeling the non-terminals of $\mathcal{G}$ to $A_0, A_1, \dots, A_{n-1}$, with $S = A_0$.

Now we define an integer-pair grammar $\mathcal{G}' = (\Sigma', \Gamma', \Delta', S')$ such that $L(\mathcal{G}') = L(\mathcal{G})$.
Let $\Sigma' = \Sigma$ and $\Gamma' = \{(a, d) : 0 \leq a,d < n \}$.
For each production rule $A_i \rightarrow A_j A_k$ in $\Delta$, add to $\Delta'$ the set of rules $(a, d) \rightarrow (a, b) (c, d)$, with $0 \leq a < n$, $d = (i - a) \bmod n$, $b = (j - a) \bmod n$, and $c = (k - d) \bmod n$.
For each production rule $A_i \rightarrow t$ in $\Delta$, add to $\Delta'$ the set of rules $(a, d) \rightarrow t$, with $0 \leq a < n$ and $d = (i - a) \bmod n$.
Let $S' = (0, 0)$.

We claim that a partial derivation $P'$ of $\mathcal{G}'$ exists if and only if the partial derivation $P$ obtained by replacing each non-terminal $(a, d)$ in $P'$ with $A_{(a + d) \bmod n}$ is a partial derivation of $\mathcal{G}$.
By construction, a rule $(a, d) \rightarrow (a, b) (c, d)$ is in $\Delta'$ if and only if the rule $A_{(a + d) \bmod n} \rightarrow A_{(a + b) \bmod n} A_{(c + d) \bmod n}$ is in $\Delta$.
Similarly, a rule $(a, d) \rightarrow t$ is in $\Delta'$ if and only if the rule $A_{(a + d) \bmod n} \rightarrow r$ is in $\Delta$.
Also, $S' = (0, 0)$ and $S = A_{(0 + 0) \bmod n}$.
So the claim holds by induction.

Since the set of all partial derivations of $P'$ are equal to those of $P$, the completed derivations are as well and $L(\mathcal{S}') = L(\mathcal{S})$. 
So $\mathcal{G}'$ expresses $\mathcal{G}$.
\qed
\end{proof}

\begin{lemma}
\label{lem:IS-express-PG}
For every integer-pair grammar $\mathcal{G}$, there exists an insertion system that expresses $\mathcal{G}$. 
\end{lemma}

\begin{proof}
Let $\mathcal{G} = (\Sigma, \Gamma, \Delta, S)$.
The integer-pair grammar $\mathcal{G}$ is expressed by an insertion system $\mathcal{S} = (\Sigma', \Delta', Q', R')$ that we now define.
Let $\Sigma' = \{s_a, s_b : (a, b) \in \Gamma\} \cup \{u, x\} \cup \Sigma$.
Let $\Delta' = \Delta_1' \cup \Delta_2' \cup \Delta_3' \cup \Delta_4'$, where
$$\Delta_1' = \{(s_b, u, s_b^*, x)^- : (a, d) \rightarrow (a, b) (c, d) \in \Delta \}$$
$$\Delta_2' = \{(s_a, s_b, s_c^*, s_d^*)^+ : (a, d) \rightarrow (a, b) (c, d) \in \Delta \}$$
$$\Delta_3' = \{(x, s_c, u^*, s_c^*)^- : (a, d) \rightarrow (a, b) (c, d) \in \Delta \}$$
$$\Delta_4' = \{(s_a, t, x, s_d^*)^+ : (a, d) \rightarrow t \in \Delta \}$$

We give each monomer type equal concentration, although the precise concentrations are not important for expressive power.
Let $Q' = (u^*, a^*)$ and $R' = (b, u)$, where $S = (a, b)$.

\textbf{Insertion types.} We start by proving that for any polymer constructed by $\mathcal{S}$, only the following types of insertions of a monomer $m_2$ between two adjacent monomers $m_1 m_3$ are possible:

\begin{enumerate}
\item $m_1 \in \Delta_2'$, $m_2 \in \Delta_3'$, $m_3 \in \Delta_1'$
\item $m_1 \in \Delta_3'$, $m_2 \in \Delta_2' \cup \Delta_4'$, $m_3 \in \Delta_1'$
\item $m_1 \in \Delta_3'$, $m_2 \in \Delta_1'$, $m_3 \in \Delta_2'$
\end{enumerate}

Moreover, we claim that for every adjacent $m_1 m_3$ pair satisfying one of these conditions, an insertion \emph{is} possible.
That is, there is a monomer $m_2$ that can be inserted, necessarily from the monomer subset specified. 

Consider each possible combination of $m_1 \in \Delta_i'$ and $m_3 \in \Delta_j'$, respectively, with $i, j \in \{1, 2, 3, 4\}$.
Observe that for an insertion to occur at insertion site $(a, b) (c, d)$, the symbols $\overline{a}$, $\overline{b}$, $\overline{c}$, and $\overline{d}$ must each occur on some monomer.
Then since $x^*$ and $t^*$ do not appear on any monomers, any $i, j$ with $i \in \{1, 4\}$ or $j \in \{3, 4\}$ cannot occur.
This leaves monomer pairs $(\Delta_i', \Delta_j')$ with $(i, j) \in \{(2, 1), (2, 2), (3, 1), (3, 2)\}$.

Insertion sites between $(\Delta_2', \Delta_1')$ pairs have the form $(s_c^*, s_d^*) (s_b, u)$, so an inserted monomer must have the form $(s_e, s_c, s_u^*, s_f)^-$ and is in $\Delta_3'$.
An insertion site $(s_c^*, s_d^*) (s_b, u)$ implies a rule of the form $(e, d) \rightarrow (e, f) (c, d)$ in $\Delta$, so there exists a monomer $(x, s_c, u^*, s_c^*)^- \in \Delta_3'$ that can be inserted.

Insertion sites between $(\Delta_3', \Delta_2')$ pairs have the form $(u^*, s_c^*) (s_a, s_b)$, so an inserted monomer must have the form $(\underline{~~}, u, s_b^*, \underline{~~})^-$ and thus is in $\Delta_1'$.
An insertion site $(u^*, s_c^*) (s_a, s_b)$ implies a rule of the form $(a, d) \rightarrow (a, b) (e, d)$ in $\Gamma$, so there exists a monomer $(s_b, u, s_b^*, x)^- \in \Delta_1'$ that can be inserted.

Insertion sites between $(\Delta_2', \Delta_2')$ pairs can only occur once a monomer $m_2 \in \Delta_2'$ has been inserted between a pair of adjacent monomers $m_1 m_3$ with either $m_1 \in \Delta_2'$ or $m_3 \in \Delta_2'$, but not both.
But we just proved that all such such possible insertions only permit $m_2 \in \Delta_3' \cup \Delta_1'$. 
Moreover, the initial insertion site between $Q'$ and $R'$ has the form $(u^*, s_a^*) (s_b, u)$ of an insertion site with $m_1 \in \Delta_3'$ and $m_3 \in \Delta_1'$.
So no pair of adjacent monomers $m_1 m_3$ are ever both from $\Delta_2'$ and no insertion site between $(\Delta_2', \Delta_2')$ pairs can ever exist.

Insertion sites between $(\Delta_3', \Delta_1')$ pairs have the form $(u^*, s_c^*) (s_b, u)$, so an inserted monomer must have the form $(s_c, \underline{~~}, \underline{~~}, b^*)^+$ or $(\underline{~~}, u, u^*, \underline{~~})^-$ and is in $\Delta_2'$ or $\Delta_4'$.
We show by induction that for each such insertion site $(u^*, s_c^*) (s_b, u)$ that $(c, b) \in \Gamma$.
First, observe that this is true for the insertion site $(u^*, s_a^*) (s_b, u)$ between $Q'$ and $R'$, since $(a, b) = S \in \Gamma$. 
Next, suppose this is true for all insertion sites of some polymer and a monomer $m_2 \in \Delta_2' \cup \Delta_4'$ is about to be inserted into the polymer between monomers from $\Delta_3'$ and $\Delta_1'$.
Inserting a monomer $m_2 \in \Delta_4'$ only reduces the set of insertion sites between monomers in $\Delta_3'$ and $\Delta_1'$, and the inductive hypothesis holds.
Inserting a monomer $m_2 \in \Delta_2'$ induces new $(\Delta_3', \Delta_2')$ and $(\Delta_2', \Delta_1')$ insertion site pairs between $m_1 m_2$ and $m_2 m_3$.
These pairs must accept two monomers $m_4 \in \Delta_1$ and $m_5 \in \Delta_3$, inducing a sequence of monomers $m_1 m_4 m_2 m_5 m_3$ with adjacent pairs $(\Delta_3', \Delta_1')$, $(\Delta_1', \Delta_2')$, $(\Delta_2', \Delta_3')$, $(\Delta_3', \Delta_1')$.
Only the first and last pairs permit insertion and both are $(\Delta_3', \Delta_1')$ pairs.

Now consider the details of the three insertions yielding $m_1 m_4 m_2 m_5 m_3$, starting with $m_1 m_3$.
The initial insertion site $m_1 m_3$ must have the form $(u^*, s_a^*) (s_d, u)$.
So the sequence of insertions has the following form, with the last two insertions interchangeable.
The symbol $\diamond$ is used to indicate the site being modified and the inserted monomer shown in bold:

$$(u^*, s_a^*) \diamond (s_d, u)$$
$$(u^*, s_a^*) \diamond \bm{(s_a, s_b, s_c^*, s_d^*)} (s_d, u)$$
$$(u^*, s_a^*) \bm{(s_b, u, s_b^*, x)} (s_a, s_b, s_c^*, s_d^*) \diamond (s_d, u)$$
$$(u^*, s_a^*) \diamond (s_b, u, s_b^*, x) (s_a, s_b, s_c^*, s_d^*) \bm{(x, s_c, u^*, s_c^*)} \diamond (s_d, u)$$

Notice the two resulting $(\Delta_3', \Delta_1')$ pair insertion sites $(u^*, s_a^*) (s_b, u)$ and $(u^*, s_c^*) (s_d, u)$.
Assume, by induction, that the monomer $m_2$ must exist.
So there is a rule $(a, d) \rightarrow (a, b) (c, d) \in \Delta$ and $(a, b), (c, d) \in \Gamma$, fulfilling the inductive hypothesis.
So for every insertion site $(u^*, s_c^*) (s_b, u)$ between a $(\Delta_3', \Delta_1')$ pair there exists a non-terminal $(c, b) \in \Gamma$.
So for every adjacent monomer pair $m_1 m_3$ with $m_1 \in \Delta_3'$ and $m_3 \in \Delta_1'$, there exists a monomer $m_2 \in \Delta_2' \cup \Delta_4'$ that can be inserted between $m_1$ and $m_2$. 

\textbf{Partial derivations and terminal polymers.} Next, consider the sequence of insertion sites between $(\Delta_3', \Delta_1')$ pairs in a polymer constructed by a modified version of $\mathcal{S}$ lacking the monomers of $\Delta_4'$.
We claim that a polymer with a sequence $(u^*, s_{a_1}^*) (s_{b_1}, u), (u^*, s_{a_2}^*) (s_{b_2}, u), \dots, (u^*, s_{a_i}^*) (s_{b_i}, u)$ of $(\Delta_3', \Delta_1')$ insertion sites is constructed if and only if there is a partial derivation $(a_1, b_1) (a_2, b_2) \dots (a_i, b_i)$ of a string in $L(\mathcal{G})$.
This follows directly from the previous proof by observing that two new adjacent $(\Delta_3', \Delta_1')$ pair insertion sites $(u^*, s_a^*) (s_b, u)$ and $(u^*, s_c^*) (s_d, u)$ can replace a $(\Delta_3', \Delta_1')$ pair insertion site if and only if there exists a rule $(a, d) \rightarrow (a, b) (c, d) \in \Delta$.

Observe that any string in $L(\mathcal{G})$ can be derived by first deriving a partial derivation containing only non-terminals, then applying only rules of the form $(a, d) \rightarrow t$.
Similarly, since the monomers of $\Delta_4'$ never form half of a valid insertion site, any terminal polymer of $\mathcal{S}$ can be constructed by first generating a polymer containing only monomers in $\Delta_1' \cup \Delta_2' \cup \Delta_3'$, then only inserting monomers from $\Delta_4'$.
Also note that the types of insertions possible in $\mathcal{S}$ imply that in any terminal polymer, any triple of adjacent monomers $m_1 m_2 m_3$ with $m_1 \in \Delta_i'$, $m_2 \in \Delta_j'$, and $m_3 \in \Delta_k'$, that $(i, j, k) \in \{(4, 1, 2), (1, 2, 3), (2, 3, 4), (3, 4, 1)\}$, with the first and last monomers of the polymer in $\Delta_4'$.

\textbf{Expression.} Define the following piecewise function $g : \Sigma' \cup \{ s^* : s \in \Sigma' \} \rightarrow \Sigma \cup \{ \varepsilon \}$ that maps to $\varepsilon$ except for second symbols of monomers in $\Delta_4'$.

\begin{displaymath}
   g(s) = \left\{
     \begin{array}{ll}
       t, & \text{if } t \in \Sigma \\
       \varepsilon, & \text{otherwise}
     \end{array}
   \right.
\end{displaymath}

Observe that every string in $L(\mathcal{S})$ has length $2 + 4 \cdot (4n - 3) + 2 = 16n-8$ for some $n \geq 0$.
Also, for each string $s_1' s_2' \dots s_{16n-8}' \in L(\mathcal{S})$, $g(s_1') g(s_2') \dots g(s_{16n-8}') =  \varepsilon^3 t_1 \varepsilon^{16} t_2 \varepsilon^{16} \dots t_n \varepsilon^5$. 
There is a terminal polymer with string representation in $L(\mathcal{S})$ yielding the sequence $s_1 s_2 \dots s_n$ if and only if the polymer can be constructed by first generating a terminal polymer excluding $\Delta_4'$ monomers with a sequence of $(\Delta_3', \Delta_1')$ insertion pairs $(a_1, b_1) (a_2, b_2) \dots (a_n, b_n)$ followed by a sequence of insertions of monomers from $\Delta_4'$ with second symbols $t_1 t_2 \dots t_n$.
Such a generation is possible if and only if $(a_1, b_1) (a_2, b_2) \dots (a_n, b_n)$ is a partial derivation of a string in $L(\mathcal{G})$ and $(a_1, b_1) \rightarrow t_1, (a_2, b_2) \rightarrow t_2, \dots, (a_n, b_n) \rightarrow t_n \in \Delta$. 
So applying the function $g$ to the string representations of the terminal polymers of $\mathcal{S}$ gives $L(\mathcal{G})$, i.e. $L(\mathcal{S}) = L(\mathcal{G})$.
Moreover, the second symbol in every fourth monomer in a terminal polymer of $\mathcal{S}$ maps to a symbol of $\Sigma$ using $g$. 
So $\mathcal{S}$ expresses $\mathcal{G}$ with the function $g$ and $\kappa = 16$.
\qed
\end{proof}

\section{Positive Results for Polymer Growth}
\label{sec:positive-results}

Dabby and Chen also consider the size and speed of constructing finite polymers.
They give a construction achieving the following result:
 
\begin{theorem}[\cite{Dabby-2013a}] 
\label{thm:dabby-chen-fast}
For any positive integer $r$, there exists an insertion system with $O(r^2)$ monomer types that deterministically constructs a polymer of length $n = 2^{\Theta(r)}$ in $O(\log^3{n})$ expected time.
Moreover, the expected time has an exponentially decaying tail probability.
\end{theorem}

We begin this section by improving on this construction significantly in both polymer length and expected running time (Theorem~\ref{thm:types-extreme-ub}).
In Section~\ref{sec:negative-results} we prove that our construction is the best possible with respect to both the polymer length and construction time for deterministic systems.

\begin{theorem}
\label{thm:types-extreme-ub}
For any positive integer $r$, there exists an insertion system with $O(r^2)$ monomer types that deterministically constructs a polymer of length~$n = 2^{\Theta(r^3)}$ in $O(\log^{5/3}(n))$ expected time.
Moreover, the expected time has an exponentially decaying tail probability.
\end{theorem}

\begin{proof}
The approach is to implement a three variable counter where each variable ranges over the values $0$ to $r$, effectively carrying out the execution of a triple for-loop. 
Insertion sites of the form $(s_a, s_b) (s_c, s_a^*)$ are used to encode the state of the counter, where $a$, $b$, and $c$ are the variables of the outer, inner, and middle loops, respectively. 
Three types of variable increments are carried out by the counter:

\begin{enumerate}[leftmargin=2cm] \itemsep5pt
\item[Inner:] If $b < r$, then $(s_a, s_b) (s_c, s_a^*) \leadsto (s_a, s_{b+1}) (s_c, s_a^*)$.
\item[Middle:] If $b = r$ and $c < r$, then $(s_a, s_b) (s_c, s_a^*) \leadsto (s_a, s_0) (s_{c+1}, s_a^*)$.
\item[Outer:] If $b = c = r$ and $a < r$, then $(s_a, s_b) (s_c, s_a^*) \leadsto (s_{a+1}, s_0) (s_0, s_{a+1}^*)$.
\end{enumerate} 

A site is \emph{modified} by a sequence of monomer insertions that yields a new usable site where all other sites created by the insertion sequence are unusable.
For instance, we modify a site $(s_a, \bm{s_b}) (s_c, s_a^*)$ to become $(s_a, \bm{s_d}) (s_c, s_a^*)$, written $(s_a, s_b) (s_c, s_a^*) \leadsto (s_a, s_d) (s_c, s_a^*)$, by adding the monomer types $(s_b^*, x, u, s_c^*)^+$ and $(x, u^*, s_a, s_d)^-$ to the system, where $x$ is a special symbol whose complement is not found on any monomer.
These two monomer types cause the following sequence of insertions, using $\diamond$ to indicate the site being modified and the inserted monomer shown in bold:
$$ (s_a, s_b) \diamond (s_c, s_a^*) $$
$$ (s_a, s_b) \bm{(s_b^*, x, u, s_c^*)} \diamond (s_c, s_a^*) $$
$$ (s_a, s_b) (s_b^*, x, u, s_c^*) \bm{(x, u^*, s_a, s_d)} \diamond (s_c, s_a^*) $$

We call this simple modification, where a single symbol in the insertion site is replaced with another symbol, a \emph{replacement}. 
There are four types of replacements, seen in Table~\ref{tab:replacements}, that can each be implemented by a pair of corresponding monomers.

\renewcommand{\arraystretch}{1.25}

\begin{table}[ht!]
\begin{center}
\begin{tabular}{| c | c |}
\hline
Replacement & Monomers \\
\hline
$(s_a, \bm{s_b}) (s_c, s_a^*) \leadsto (s_a, \bm{s_d}) (s_c, s_a^*)$ & $(s_b^*, x, u, s_c^*)^+$, $(x, u^*, s_a, s_d)^-$   \\
$(s_a, s_b) (\bm{s_c}, s_a^*) \leadsto (s_a, s_b) (\bm{s_d}, s_a^*)$ & $(s_b^*, u, x, s_c^*)^+$, $(s_d, s_a^*, u^*, x)^-$ \\
$(\bm{s_b}, s_a) (s_a^*, s_c) \leadsto (\bm{s_d}, s_a) (s_a^*, s_c)$ & $(x, s_b^*, s_c^*, u)^-$, $(u^*, x, s_d, s_a)^+$   \\
$(s_b, s_a) (s_a^*, \bm{s_c}) \leadsto (s_b, s_a) (s_a^*, \bm{s_d})$ & $(u, s_b^*, s_c^*, x)^-$, $(s_a^*, s_d, x, u^*)^+$ \\ 
\hline
\end{tabular}
\end{center}
\caption{The four types of replacement steps and monomer pairs that implement them.
The symbol $u$ can be any symbol, and $x$ is a special symbol whose complement does not appear on any monomer.}
\label{tab:replacements}
\end{table}

Each of the three increment types are implemented using a sequence of site modifications.
The resulting triple for-loop carries out a sequence of $\Theta(r^3)$ insertions, constructing a $\Theta(r^3)$-length polymer.
A $2^{\Theta(r^3)}$-length polymer is achieved by simultaneously duplicating each site during each inner increment.
Because the for-loop runs for $\Theta(r^3)$ steps and duplicates at a constant fraction of these steps (those with $0 \leq b < r$), the number of counters reaching the final $a = b = c = r$ state is $2^{\Theta(r^3)}$.
In the remainder of the proof, we detail the implementation of each increment type, starting with the simplest: middle increments.

\textbf{Middle increment.}
A middle increment of a site $(s_a, s_b) (s_c, s_a^*)$ occurs when the site has the form $(s_a, s_r) (s_c, s_a^*)$ with $0 \leq c < r$, performing the modification $(s_a, s_r) (s_c, s_a^*) \leadsto (s_a, s_0) (s_{c+1}, s_a^*)$.
We implement middle increments using a sequence of three replacements:
$$ (s_a, s_r) (s_c, s_a^*) \overset{1}{\leadsto} (s_a, s_r) (s_{f_1(c)}, s_a^*) \overset{2}{\leadsto} (s_a, s_0) (s_{f_1(c)}, s_a^*) \overset{3}{\leadsto} (s_a, s_0) (s_{c+1}, s_a^*) $$

where $f_i(n) = n + 2ir^2$.
The use of $f$ is to avoid unintended interactions between monomers, since for any $n_1, n_2$ with $0 \leq n_1, n_2 \leq r$, $f_i(n_1) \neq f_j(n_2)$ for all $i \neq j$.
Compiling this sequence of replacements into monomer types gives the following monomers:

\begin{enumerate}[label=Step \arabic*:, leftmargin=2cm]
\item $(s_r^*, s_{f_2(c)}, x, s_c^*)^+$ and $(s_{f_1(c)}, s_a^*, s_{f_2(c)}^*, x)^-$.
\item $(s_r^*, x, s_{f_3(c)}, s_{f_1(c)}^*)^+$ and $(x, s_{f_3(c)}^*, s_a, s_0)^-$.
\item $(s_0^*, s_{f_4(c+1)}, x, s_{f_1(c)}^*)^+$ and $(s_{c+1}, s_a^*, s_{f_4(c+1)}^*, x)^-$. 
\end{enumerate}

This set of monomers results in the following sequence of six insertions:
$$ (s_a, s_r) \diamond (s_c, s_a^*) $$
$$ (s_a, s_r) \diamond \bm{(s_r^*, s_{f_2(c)}, x, s_c^*)} (s_c, s_a^*) $$
$$ (s_a, s_r) \diamond \bm{(s_{f_1(c)}, s_a^*, s_{f_2(c)}^*, x)} (s_r^*, s_{f_2(c)}, x, s_c^*) (s_c, s_a^*) $$
$$ (s_a, s_r) \diamond (s_{f_1(c)}, s_a^*) $$
$$ (s_a, s_r) \bm{(s_r^*, x, s_{f_3(c)}, s_{f_1(c)}^*)} \diamond (s_{f_1(c)}, s_a^*) $$
$$ (s_a, s_r) (s_r^*, x, s_{f_3(c)}, s_{f_1(c)}^*) \bm{(x, s_{f_3(c)}^*, s_a, s_0)} \diamond (s_{f_1(c)}, s_a^*) $$
$$ (s_a, s_0) \diamond (s_{f_1(c)}, s_a^*) $$
$$ (s_a, s_0) \diamond \bm{(s_0^*, s_{f_4(c+1)}, x, s_{f_1(c)}^*)} (s_{f_1(c)}, s_a^*) $$
$$ (s_a, s_0) \diamond \bm{(s_{c+1}, s_a^*, s_{f_4(c+1)}^*, x)} (s_0^*, s_{f_4(c+1)}, x, s_{f_1(c)}^*) (s_{f_1(c)}, s_a^*) $$
$$ (s_a, s_0) \diamond (s_{c+1}, s_a^*) $$

Since each inserted monomer has an instance of $x$, all other insertion sites created are unusable.
This is true of the insertions used for outer increments and duplications as well.

\textbf{Outer increment.}
An outer increment of the site $(s_a, s_b) (s_c, s_a^*)$ occurs when the site has the form $(s_a, s_r) (s_r, s_a^*)$ with $0 \leq a < r$.
We implement this step using a two-phase sequence of three normal replacements and a special quadruple replacement (Step 3):
$$ (s_a, s_r) (s_r, s_a^*) \overset{1}{\leadsto} (s_a, s_{f_5(a)}) (s_r, s_a^*) \overset{2}{\leadsto} (s_a, s_{f_5(a)}) (s_{f_5(a)}^*, s_a^*) $$ 
$$ (s_a, s_{f_5(a)}) (s_{f_5(a)}^*, s_a^*) \overset{3}{\leadsto} (s_{a+1}, s_{f_5(0)}) (s_0, s_{a+1}^*) \overset{4}{\leadsto} (s_{a+1}, s_0) (s_0, s_{a+1}^*) $$ 

At each step, a (necessary) complementary pair of symbols is maintained, which results in a sequence of more than~4 replacements.
As with inner and middle increments, we compile replacement steps~1,~2, and~4 into monomers using Table~\ref{tab:replacements}.
Step~3 is a special pair of monomers.

\begin{enumerate}[label=Step \arabic*:, leftmargin=2cm]
\item $(s_r^*, x, s_{f_6(r)}, s_r^*)^+$ and $(x, s_{f_6(r)}^*, s_a, s_{f_5(a)})^-$.
\item $(s_{f_5(a)}^*, s_{f_7(r)}^*, x, s_r^*)^+$ and $(s_{f_5(a)}^*, s_a^*, s_{f_7(r)}, x)^-$. 
\item $(s_{f_5(a)}^*, x, s_{a+1}, s_{f_5(a)})^+$ and $(s_0, s_{a+1}^*, s_a, x)^-$.
\item $(s_{f_5(a)}^*, x, s_{f_7(r)}, s_0^*)^+$ and $(x, s_{f_7(r)}^*, s_{a+1}, s_0)^-$.  
\end{enumerate} 

Here is the sequence of insertions, using $\diamond$ to indicate the site being modified and the inserted monomer shown in bold:
$$ (s_a, s_r) \diamond (s_r, s_a^*) $$ 
$$ (s_a, s_r) \bm{(s_r^*, x, s_{f_6(r)}, s_r^*)} \diamond (s_r, s_a^*) $$ 
$$ (s_a, s_r) (s_r^*, x, s_{f_6(r)}, s_r^*) \bm{(x, s_{f_6(r)}^*, s_a, s_{f_5(a)})} \diamond (s_r, s_a^*) $$ 
$$ (s_a, s_{f_5(a)}) \diamond (s_r, s_a^*) $$ 
$$ (s_a, s_{f_5(a)}) \diamond \bm{(s_{f_5(a)}^*, s_{f_7(r)}^*, x, s_r^*)} (s_r, s_a^*) $$ 
$$ (s_a, s_{f_5(a)}) \diamond \bm{(s_{f_5(a)}^*, s_a^*, s_{f_7(r)}, x)} (s_{f_5(a)}^*, s_{f_7(r)}^*, x, s_r^*) (s_r, s_a^*) $$ 
$$ (s_a, s_{f_5(a)}) \diamond (s_{f_5(a)}^*, s_a^*) $$
$$ (s_a, s_{f_5(a)}) \bm{(s_{f_5(a)}^*, x, s_{a+1}, s_{f_5(a)})} \diamond (s_{f_5(a)}^*, s_a^*) $$
$$ (s_a, s_{f_5(a)}) (s_{f_5(a)}^*, x, s_{a+1}, s_{f_5(a)}) \diamond \bm{(s_0, s_{a+1}^*, s_a, x)} (s_{f_5(a)}^*, s_a^*) $$
$$ (s_{a+1}, s_{f_5(a)}) \diamond (s_0, s_{a+1}^*) $$
$$ (s_{a+1}, s_{f_5(a)}) \bm{(s_{f_5(a)}^*, x, s_{f_7(r)}, s_0^*)} \diamond (s_0, s_{a+1}^*) $$
$$ (s_{a+1}, s_{f_5(a)}) (s_{f_5(a)}^*, x, s_{f_7(r)}, s_0^*) \bm{(x, s_{f_7(r)}^*, s_{a+1}, s_0)} \diamond (s_0, s_{a+1}^*) $$
$$ (s_{a+1}, s_0) \diamond (s_0, s_{a+1}^*) $$

\textbf{Inner increment.}
The inner increment has two phases.
The first phase (Steps~1-2) performs duplication, modifying the initial site to a pair of sites: $(s_a, s_b) (s_c, s_a^*) \leadsto (s_a, s_b) (s_{f_8(c)}, s_a^*) \dots (s_a, s_{b+1}) (s_c, s_a^*)$, yielding an incremented version of the original site and one other site.
The second phase (Steps~3-5) is $(s_a, s_b) (s_{f_8(c)}, s_a^*) \leadsto (s_a, s_{b+1}) (s_c, a^*)$, transforming the second site into an incremented version of the original site.

For the first phase, we use the three monomers:

\begin{enumerate}[leftmargin=2cm]
\item[Step 1:] $(s_b^*, s_{f_8(c)}, s_{f_8(b+1)}, s_c^*)^+$.
\item[Step 2:] $(s_{f_8(c)}, s_a^*, s_{f_8(c)}^*, x)^-$ and $(x, s_{f_8(b+1)}^*, s_a, s_{b+1})^-$.
\end{enumerate}

Here is the sequence of insertions:
$$ (s_a, s_b) \diamond (s_c, s_a^*) $$
$$ (s_a, s_b) \diamond \bm{(s_b^*, s_{f_8(c)}, s_{f_8(b+1)}, s_c^*)} \diamond (s_c, s_a^*) $$ 
$$ (s_a, s_b) \diamond \bm{(s_{f_8(c)}, s_a^*, s_{f_8(c)}^*, x)} (s_b^*, s_{f_8(c)}, s_{f_8(b+1)}, s_c^*) \diamond (s_c, s_a^*) $$
$$ (s_a, s_b) \diamond (s_{f_8(c)}, s_a^*, s_{f_8(c)}^*, x) (s_b^*, s_{f_8(c)}, s_{f_8(b+1)}, s_c^*) \bm{(x, s_{f_8(b+1)}^*, s_a, s_{b+1})} \diamond (s_c, s_a^*) $$
$$ (s_a, s_b) \diamond (s_{f_8(c)}, s_a^*) \dots (s_a, s_{b+1}) \diamond (s_c, s_a^*) $$

The last two insertions occur independently and may happen in the opposite order of the sequence depicted here.
In the second phase, the site $(s_a, s_b) (s_{f_8(c)}, s_a^*)$ is transformed into $(s_a, s_{b+1}) (s_c, s_a^*)$ by a sequence of replacement steps:
$$ (s_a, s_b) (s_{f_8(c)}, s_a^*) \overset{3}{\leadsto} (s_a, s_{f_9(b)}) (s_{f_8(c)}, s_a^*) \overset{4}{\leadsto} (s_a, s_{f_9(b)}) (s_c, s_a^*) \overset{5}{\leadsto} (s_a, s_{b+1}) (s_c, s_a^*) $$ 

As with previous sequences of replacement steps, we compile this sequence into a set of monomers:

\begin{enumerate}[leftmargin=2cm]
\item[Step 3:] $(s_b^*, x, s_{f_{10}(b)}, s_{f_8(c)}^*)^+$ and $(x, s_{f_{10}(b)}^*, s_a, s_{f_9(b)})^-$. 
\item[Step 4:] $(s_{f_9(b)}^*, s_{f_{11}(c)}, x, s_{f_8(c)}^*)^+$ and $(s_c, s_a^*, s_{f_{11}(c)}^*, x)^-$. 
\item[Step 5:] $(s_{f_9(b)}^*, x, s_{f_{12}(b+1)}, s_c^*)^+$ and $(x, s_{f_{12}(b+1)}^*, s_a, s_{b+1})^-$.
\end{enumerate}

Here is the resulting sequence of insertions:
$$ (s_a, s_b) \diamond (s_{f_8(c)}, s_a^*) $$
$$ (s_a, s_b) \bm{(s_b^*, x, s_{f_{10}(b)}, s_{f_8(c)}^*)} \diamond (s_{f_8(c)}, s_a^*) $$
$$ (s_a, s_b) (s_b^*, x, s_{f_{10}(b)}, s_{f_8(c)}^*) \bm{(x, s_{f_{10}(b)}^*, s_a, s_{f_9(b)})} \diamond (s_{f_8(c)}, s_a^*) $$
$$ (s_a, s_{f_9(b)}) \diamond (s_{f_8(c)}, s_a^*) $$
$$ (s_a, s_{f_9(b)}) \diamond \bm{(s_{f_9(b)}^*, s_{f_{11}(c)}, x, s_{f_8(c)}^*)} (s_{f_8(c)}, s_a^*) $$
$$ (s_a, s_{f_9(b)}) \diamond \bm{(s_c, s_a^*, s_{f_{11}(c)}^*, x)} (s_{f_9(b)}^*, s_{f_{11}(c)}, x, s_{f_8(c)}^*) (s_{f_8(c)}, s_a^*) $$
$$ (s_a, s_{f_9(b)}) \diamond (s_c, s_a^*) $$
$$ (s_a, s_{f_9(b)}) \bm{(s_{f_9(b)}^*, x, s_{f_{12}(b+1)}, s_c^*)} \diamond (s_c, s_a^*) $$
$$ (s_a, s_{f_9(b)}) (s_{f_9(b)}^*, x, s_{f_{12}(b+1)}, s_c^*) \bm{(x, s_{f_{12}(b+1)}^*, s_a, s_{b+1})} \diamond (s_c, s_a^*) $$
$$ (s_a, s_{b+1}) \diamond (s_c, s_a^*) $$

When combined, the two phases of duplication modify $(s_a, s_b) (s_c, s_a^*)$ to become $(s_a, s_{b+1}) (s_c, s_a^*) \dots (s_a, s_{b+1}) (s_c, s_a^*)$, where all sites between the duplicated sites are unusable.
Notice that although we need to duplicate $\Theta(r^3)$ distinct sites, only $\Theta(r^2)$ monomers are used in the implementation since each monomer either does not depend on $a$, e.g. $(s_b^*, x, s_{f_{10}(b)}, s_{f_8(c)}^*)^+$, or does not depend on $c$, e.g. $(x, s_{f_{10}(b)}^*, s_a, s_{f_9(b)})^-$.

\renewcommand{\arraystretch}{1.25}

\begin{table}
\begin{center}
\begin{tabular}{| c | l l |}
\hline
Step     & \multicolumn{2}{|c|}{Inner monomer types ($b < r$)}                        \\
\hline
1        & \multicolumn{2}{|c|}{$(s_b^*, s_{f_8(c)}, s_{f_8(b+1)}, s_c^*)^+$}                             \\ 
2        & $(s_{f_8(c)}, s_a^*, s_{f_8(c)}^*, x)^-$           & $(x, s_{f_8(b+1)}^*, s_a, s_{b+1})^-$     \\ 
3        & $(s_b^*, x, s_{f_{10}(b)}, s_{f_8(c)}^*)^+$        & $(x, s_{f_{10}(b)}^*, s_a, s_{f_9(b)})^-$ \\
4        & $(s_{f_9(b)}^*, s_{f_{11}(c)}, x, s_{f_8(c)}^*)^+$ & $(s_c, s_a^*, s_{f_{11}(c)}^*, x)^-$      \\
5        & $(s_{f_9(b)}^*, x, s_{f_{12}(b+1)}, s_c^*)^+$      & $(x, s_{f_{12}(b+1)}^*, s_a, s_{b+1})^-$  \\ [3pt]
\hline
Step     & \multicolumn{2}{|c|}{Middle monomer types ($c < r$)} \\
\hline
1        & $(s_r^*, s_{f_2(c)}, x, s_c^*)^+$          & $(s_{f_1(c)}, s_a^*, s_{f_2(c)}^*, x)^-$  \\
2        & $(s_r^*, x, s_{f_3(c)}, s_{f_1(c)}^*)^+$   & $(x, s_{f_3(c)}^*, s_a, s_0)^-$           \\
3        & $(s_0^*, s_{f_4(c+1)}, x, s_{f_1(c)}^*)^+$ & $(s_{c+1}, s_a^*, s_{f_4(c+1)}^*, x)^-$   \\ [3pt]
\hline
Step     & \multicolumn{2}{|c|}{Outer monomer types ($a < r$)} \\
\hline
1        & $(s_r^*, x, s_{f_6(r)}, s_r^*)^+$          & $(x, s_{f_6(r)}^*, s_a, s_{f_5(a)})^-$   \\
2        & $(s_{f_5(a)}^*, s_{f_7(r)}^*, x, s_r^*)^+$ & $(s_{f_5(a)}^*, s_a^*, s_{f_7(r)}, x)^-$ \\
3        & $(s_{f_5(a)}^*, x, s_{a+1}, s_{f_5(a)})^+$ & $(s_0, s_{a+1}^*, s_a, x)^-$             \\
4        & $(s_{f_5(a)}^*, x, s_{f_7(r)}, s_0^*)^+$   & $(x, s_{f_7(r)}^*, s_{a+1}, s_0)^-$      \\ [3pt]
\hline
\end{tabular}
\end{center}
\caption{The set of all monomer types used to deterministically construct a monomer of size $2^{\Theta(r^3)}$ using $O(r^2)$ monomer types.}
\label{tab:all-monomers-types-extreme-ub}
\end{table}

\textbf{Putting it together.}
The system starts with the intiator $(s_0, s_0) (s_0, s_0^*)$.
Each increment of the counter occurs either through a middle increment, outer increment, or a duplication.
The total set of monomers is seen in Table~\ref{tab:all-monomers-types-extreme-ub}.
There are at most $(r+1)^2$ monomer types in each family (each row of Table~\ref{tab:all-monomers-types-extreme-ub}) and $O(r^2)$ monomer types total.

The system is deterministic if no pair of monomers can be inserted into any insertion site appearing during construction.
It can be verified by an inspection of Table~\ref{tab:all-monomers-types-extreme-ub} that any pair of positive monomers have a distinct pair of first and fourth symbols, and any pair of negative monomers have a distinct pair of second and third symbols.
So all insertion sites with only one pair of complementary symbols have at most one insertable monomer, and moreover create only those sites listed via the insertion sequences.

The first insertion site with two pairs of complementary symbols must be one created by a sequence of insertions into sites with a single pair of complementary symbols (that have only one insertable monomer).
So this site has the form $(s_a, s_{f_6(a)}) (s_{f_6(a)}^*, s_a^*)$ for some $0 \leq a < r$, created during the outer increment step.
For each such site, a positive monomer $(s_{f_6(a)}^*, x, s_{a+1}, s_{f_6(a)})^+$ can attach, but no negative monomer has second and third symbols $s_a^*$ and $s_a$.
So these sites have only one insertable monomer and are the only insertion sites with two pairs of complementary symbols.
So all sites with two pairs of complementary symbols also have at most one insertable monomer.
Then since all sites occurring in the system have at most one insertable monomer, the system is deterministic.

The size $P_i$ of a subpolymer with an initiator encoding some value $i$ between $0$ and $(r+1)^3-1$ can be bounded by $2P_{i+2} + 9 \leq P_i \leq 2P_{i+1} + 9$, since either $i+1$ or $i+2$ is an inner increment step and no step inserts more than 9 monomers. 
Moreover, $P_{(r+1)^3-2} \geq 1$.
So $P_0 + 2$, the size of the terminal polymer, is $2^{\Theta(r^3)}$. 
  
\textbf{Running time.}
Define the concentration of each monomer type to be equal.
There are $12r^2 + 24r + 3 \leq 39r^2$ monomer types, so each monomer type has concentration at least $1/(39r^2)$.
The polymer is complete as soon as every counter's variables have reached the value $a = b = c = r$, i.e. every site encoding a counter has been modified to become $(s_r, s_r) (s_r, s_r^*)$ and the monomer $(s_r^*, x, s_{f_7(r)}, s_r^*)^+$ has been inserted.

There are fewer than $2^{r^3}$ such insertions, and each insertion requires at most $9 \cdot (r+1)^3 \leq 72r^3$ previous insertions to occur.
So an upper bound on the expected time $T_r$ for each such insertion is described as a sum of $72r^3$ random variables, each with expected time $39r^2$.
The Chernoff bound for exponential random variables implies the following upper bound on $T_r$:

$\begin{aligned}
{\rm Prob}[T_r > 39r^2 \cdot 72r^3(1 + \delta)] &\leq e^{-39 \cdot 72r^5 \delta^2 / (2 + \delta)} \\
&\leq e^{-r^5 \delta^2 / (2 + \delta)} \\
&\leq e^{-r^5 \delta^2 / (2\delta)} \rm{~for~all~} \delta \geq 2 \\
&\leq e^{-r^5 \delta / 2} \\
\end{aligned}$

Let $T_{\mathcal{S}_r}$ be the total running time of the system. 
Then we can bound $T_{\mathcal{S}_r}$ from above using the bound for $T_r$:

$\begin{aligned}
{\rm Prob}[T_{\mathcal{S}_r} > 39r^2 \cdot 72r^3(1 + \delta)] &\leq 2^{r^3} \cdot e^{-r^5 \delta / 2} \\
&\leq 2^{r^3} 2^{-r^5 \delta/2} \\ 
&\leq 2^{r^3 - r^5 \delta/2} \\
&\leq 2^{r^5 \delta/4 - r^5 \delta/2} \text{~for~all~} \delta \geq 4 \\
&\leq 2^{-r^5\delta/4}
\end{aligned}$

So ${\rm Prob}[T_{\mathcal{S}_r} > 39r^2 \cdot 72r^3(1 + \delta)] \leq 2^{-r^5\delta/4}$ for all $\delta \geq 4$.
So the expected value of $T_{\mathcal{S}_r}$, the construction time, is $O(r^5) = O(\log^{5/3}(n))$ with an exponentially decaying tail probability.
\qed
\end{proof}

It is natural to ask whether faster construction of polymers is possible in non-deterministic systems: systems that do not construct a single terminal polymer.
A two-monomer-type insertion system consisting of the initiator $(s_1, s_2) (s_2^*, s_1^*)$ and monomer types $(s_2^*, s_1^*, s_1, s_2)^+$, and $(s_2^*, x, x, s_2)^+$ simultaneously constructs polymers of all lengths $n \geq 3$ in expected time $O(\log{n})$ via balanced insertion sequences of logarithmic length.
Moreover, any polymer in any system has $\Omega(\log{n})$ expected construction time, since every insertion takes $\Omega(1)$ expected time, and constructing a polymer of length $n$ requires an insertion sequence of length at least $\lfloor \log_2(n-2) \rfloor$. 
So if assembling anything is permitted, then this two-monomer-type system is asymptotically optimal.

What about a milder form of non-determinism: permitting system to construct a finite total number of polymers?
Our next result proves that even this relaxation is sufficient to improve construction time.
The key idea is allow large sets of monomer types to ``compete'' to insert into a common insertion site first.
This competition increases the total concentration of insertable monomer types, reducing the expected insertion time, but results in a non-deterministic system.
 
\begin{theorem}
\label{thm:speed-extreme-ub}
For any positive, odd integer $r$, there exists an insertion system constructing a finite set of polymers with $O(r^2)$ monomer types that constructs a polymer of length~$n = 2^{\Theta(r^2)}$ in $O(\log^{3/2}(n))$ expected time.
Moreover, the expected time has an exponentially decaying tail probability.
\end{theorem}

\begin{proof}
We use a construction similar to one used in the proof of Theorem~\ref{thm:types-extreme-ub}, but with a few key differences.
First, this construction carries out the execution of a double for-loop, not a triple for-loop. 
Second, each increment step involves a sequence of ``guesses'': non-deterministic insertions where many different monomer types can be inserted.
Incorrect guesses stop further execution, while 
At these insertions, only one monomer type allows continued execution of the for-loop, and all others cause the loop to break.

We use $a$ and $b$ for the outer and inner for-loop variables, and a parity bit to to encode the parity of $b$.
Insertion sites of the form $(s_{f_p(a)}, s_{f_p(b)}) (s_{f_p(b)}^*, s_{f_p(a)}^*)$ encode counter states, where $p = b \bmod 2$ and $f_i(n) = n + 2ir^2$.
The parity bit is used to avoid an issue with incorrect guesses causing repeated insertion sites (and thus an infinite set of polymers).
Three types of variable increments are carried out by the counter:

\begin{enumerate}[leftmargin=3cm] \itemsep5pt
\item[Inner even-to-odd:] If $b < r$ and $b$ is even, then
$(s_{f_0(a)}, s_{f_0(b)}) (s_{f_0(b)}^*, s_{f_0(a)}^*) \leadsto (s_{f_1(a)}, s_{f_1(b+1)}) (s_{f_1(b+1)}^*, s_{f_1(a)}^*)$.
\item[Inner odd-to-even:] If $b < r$ and $b$ is odd, then
$(s_{f_1(a)}, s_{f_1(b)}) (s_{f_1(b)}^*, s_{f_1(a)}^*) \leadsto (s_{f_0(a)}, s_{f_0(b+1)}) (s_{f_0(b+1)}^*, s_{f_0(a)}^*)$.
\item[Outer:] If $b = r$ and $a < r$, then
$(s_{f_1(a)}, s_{f_1(r)}) (s_{f_1(r)}^*, s_{f_1(a)}^*) \leadsto (s_{f_0(a+1)}, s_{f_0(0)}) (s_{f_0(0)}^*, s_{f_0(a+1)}^*)$.
\end{enumerate} 

As in the construction of Theorem~\ref{thm:types-extreme-ub}, this counter carries out a sequence of $\Theta(r^2)$ insertions, which are used to construct a length $2^{\Theta(r^2)}$ polymer by simultaneously duplicating each site during an inner even-to-odd increment.
A complete set of monomers for the system can be seen in Table~\ref{tab:all-monomers-speed-extreme-ub}.

In the remainder of the proof we detail the implementation of each increment type.
The bulk of the effort lies in proving that every insertion site accepting multiple monomer types has a single \emph{correct} monomer type that continues the execution of the loops, and all other \emph{incorrect} monomer types yield insertion sites that do not accept any monomer types.
We call such an insertion site \emph{growth-deterministic}, and if every insertion site appearing in a system is growth deterministic, the polymer constructed by only inserting correct monomer types is the longest polymer constructed by the system. 

\renewcommand{\arraystretch}{1.25}

\begin{table}
\begin{center}
\begin{tabular}{| c | l l |}
\hline
Step     & \multicolumn{2}{|c|}{Inner even-to-odd monomer types ($b < r$, $b$ even)} \\
\hline
\multirow{2}{*}{1} & \multicolumn{2}{|c|}{$(s_{f_1(b+1)}^*, s_{f_0(a)}^*, s_{f_0(a)}, x)^-$} \\
                   & $(s_{f_0(b)}^*, x, s_{f_1(a)}, s_{f_1(b+1)})^+$     & $(s_{f_1(b+1)}^*, s_{f_1(a)}^*, s_{f_0(a)}, s_{f_2(b+2)})^-$  \\
2                  & $(s_{f_2(b+2)}^*, s_{f_2(a)}^*, x, s_{f_1(b+1)})^+$ & $(s_{f_0(b+2)}^*, s_{f_0(a)}^*, s_{f_2(a)}, x)^-$ \\
3                  & $(s_{f_2(b+2)}^*, x, s_{f_3(a)}, s_{f_0(b+2)})^+$   & $(x, s_{f_3(a)}^*, s_{f_0(a)}, s_{f_0(b+2)})^-$ \\ [3pt] 
\hline
\multicolumn{3}{|c|}{Inner odd-to-even monomer types ($b < r$, $b$ odd)} \\
\hline
\multirow{2}{*}{1}  & $(s_{f_0(b+1)}^*, s_{f_1(a)}^*, s_{f_1(a)}, x)^-$ & $(s_{f_1(b)}^*, x, s_{f_0(a)}, s_{f_0(b+1)})^+$ \\
  & \multicolumn{2}{|c|}{$(s_{f_0(b+1)}^*, s_{f_0(a)}^*, s_{f_1(a)}, x)^-$} \\ [3pt] 
\hline
\multicolumn{3}{|c|}{Outer monomer types ($a < r$)} \\
\hline
1  & $(s_{f_1(r)}^*, x, s_{f_3(a+1)}, s_{f_1(r)})^+$ & $(s_{f_0(0)}^*, s_{f_3(a+1)}^*, s_{f_1(a)}, x)^-$   \\
2  & $(s_{f_1(r)}^*, x, s_{f_0(a+1)}, s_{f_0(0)})^+$ & $(s_{f_0(0)}^*, s_{f_0(a+1)}^*, s_{f_3(a+1)}, x)^-$ \\ [3pt]
\hline
\end{tabular}
\end{center}
\caption{The set of all monomer types used to construct a monomer of size $n$ in $O(\log^{3/2}(n))$ expected time while constructing only finite polymers.}
\label{tab:all-monomers-speed-extreme-ub}
\end{table}

\textbf{Inner odd-to-even increment.}
We implement inner odd-to-even increments in a single step of three insertions.
In an sequence of three correct insertions, the monomer types 
$$(s_{f_0(b+1)}^*, s_{f_1(a)}^*, s_{f_1(a)}, x)^-~~~(s_{f_1(b)}^*, x, s_{f_0(a)}, s_{f_0(b+1)})^+~~~(s_{f_0(b+1)}^*, s_{f_0(a)}^*, s_{f_1(a)}, x)^-$$ 
are used.
This set of monomers results the sequence of insertions
$$ (s_{f_1(a)}, s_{f_1(b)}) \diamond (s_{f_1(b)}^*, s_{f_1(a)}^*) $$ 
$$ (s_{f_1(a)}, s_{f_1(b)}) \diamond \bm{(s_{f_0(b+1)}^*, s_{f_1(a)}^*, s_{f_1(a)}, x)} (s_{f_1(b)}^*, s_{f_1(a)}^*) $$ 
$$ (s_{f_1(a)}, s_{f_1(b)}) \diamond (s_{f_0(b+1)}^*, s_{f_1(a)}^*) $$
$$ (s_{f_1(a)}, s_{f_1(b)}) \bm{(s_{f_1(b)}^*, x, s_{f_0(a)}, s_{f_0(b+1)})} \diamond (s_{f_0(b+1)}^*, s_{f_1(a)}^*) $$
$$ (s_{f_0(a)}, s_{f_0(b+1)}) \diamond (s_{f_0(b+1)}^*, s_{f_1(a)}^*) $$
$$ (s_{f_0(a)}, s_{f_0(b+1)}) \diamond \bm{(s_{f_0(b+1)}^*, s_{f_0(a)}^*, s_{f_1(a)}, x)} (s_{f_0(b+1)}^*, s_{f_1(a)}^*) $$
$$ (s_{f_0(a)}, s_{f_0(b+1)}) \diamond (s_{f_0(b+1)}^*, s_{f_0(a)}^*) $$

The first insertion site $(s_{f_1(a)}, s_{f_1(b)}) (s_{f_1(b)}^*, s_{f_1(a)}^*)$ accepts any monomer of the form $(\underline{~~}, s_{f_1(a)}^*, s_{f_1(a)}, \underline{~~})^-$ or $(s_{f_1(b)}^*, \underline{~~}, \underline{~~}, s_{f_1(b)})^+$.
The only such monomer types in the system are $(s_{f_0(i+1)}^*, s_{f_1(a)}^*, s_{f_1(a)}, x)^-$, where $i$ is an odd integer with $0 \leq i < r$ (inner odd-to-even monomers).
So the insertion has the form
$$ (s_{f_1(a)}, s_{f_1(b)}) \diamond (s_{f_1(b)}^*, s_{f_1(a)}^*) $$ 
$$ (s_{f_1(a)}, s_{f_1(b)}) \bm{(s_{f_0(i+1)}^*, s_{f_1(a)}^*, s_{f_1(a)}, x)} (s_{f_1(b)}^*, s_{f_1(a)}^*) $$ 

Only the left resulting site is insertable.
Moreover, this site is only insertable if $i = b$, since the only monomers of the form $(s_{f_1(b)}^*, \underline{~~}, \underline{~~}, s_{f_0(i+1)})^+$ in the system have $i = b$. 
So only one monomer type inserted into the first site enables further growth: $(s_{f_0(b+1)}^*, s_{f_1(a)}^*, s_{f_1(a)}, x)^-$.
Thus the first site is growth-deterministic.

The second insertion site accepts monomers of the form $(s_{f_1(b)}^*, \underline{~~}, \underline{~~}, s_{f_0(b+1)})^+$.
The only such monomer types in the system are $(s_{f_1(b)}^*, x, s_{f_0(i)}, s_{f_0(b+1)})^+$ with $0 \leq i \leq r$ (inner odd-to-even monomers). 
So the insertion has the form
$$ (s_{f_1(a)}, s_{f_1(b)}) \diamond (s_{f_0(b+1)}^*, s_{f_1(a)}^*) $$
$$ (s_{f_1(a)}, s_{f_1(b)}) \bm{(s_{f_1(b)}^*, x, s_{f_0(i)}, s_{f_0(b+1)})} (s_{f_0(b+1)}^*, s_{f_1(a)}^*) $$

Only the right resulting site is insertable.
Moreover, this site is only insertable if $i = a$, since the only monomers of the form $(\underline{~~}, s_{f_0(i)}^*, s_{f_1(a)}, \underline{~~})^-$ in the system have $i = a$.
So only one monomer inserted into the second site enables further growth: $(s_{f_1(b)}^*, x, s_{f_0(a)}, s_{f_0(b+1)})^+$.
So the second site is growth-deterministic. 

The third insertion site accepts monomers of the form $(\underline{~~}, s_{f_0(a)}^*, s_{f_1(a)}, \underline{~~})^-$.
The only such monomer types in the system are $(s_{f_0(i+1)}^*, s_{f_0(a)}^*, s_{f_1(a)}, x)^-$ where $i$ is an odd integer with $0 \leq i < r$ (inner odd-to-even monomers). 
So the insertion has the form
$$ (s_{f_0(a)}, s_{f_0(b+1)}) \diamond (s_{f_0(b+1)}^*, s_{f_1(a)}^*) $$
$$ (s_{f_0(a)}, s_{f_0(b+1)}) \bm{(s_{f_0(i+1)}^*, s_{f_0(a)}^*, s_{f_1(a)}, x)} (s_{f_0(b+1)}^*, s_{f_1(a)}^*) $$

Only the left resulting site is insertable.
Moreover, this site is only insertable if $i = b$, since the only monomers of the form $(s_{f_0(b+1)}, \underline{~~}, \underline{~~}, s_{f_0(i+1)}^*)^+$ or $(\underline{~~}, s_{f_0(a)}^*, s_{f_0(a)}, \underline{~~})^-$ in the system are the latter, which require $i = b$ to insert.
So only one monomer inserted into the second site enables further growth: $(s_{f_0(b+1)}^*, s_{f_0(a)}^*, s_{f_1(a)}, x)^-$.
So the third site is growth-deterministic.
\vspace{2px} 

\textbf{Outer increment.}
We implement outer increments in two steps:
$$ (s_{f_1(a)}, s_{f_1(r)}) (s_{f_1(r)}^*, s_{f_1(a)}^*) \overset{1}{\leadsto} (s_{f_3(a+1)}, s_{f_1(r)}) (s_{f_0(0)}^*, s_{f_3(a+1)}^*) $$
$$ (s_{f_3(a+1)}, s_{f_1(r)}) (s_{f_0(0)}^*, s_{f_3(a+1)}^*) \overset{2}{\leadsto} (s_{f_0(a+1)}, s_{f_0(0)}) (s_{f_0(0)}^*, s_{f_0(a+1)}^*) $$

In a sequence of four correct insertions, the monomer types 
\begin{enumerate}[leftmargin=2cm]
\item[Step 1:] $(s_{f_1(r)}^*, x, s_{f_3(a+1)}, s_{f_1(r)})^+$ and $(s_{f_0(0)}^*, s_{f_3(a+1)}^*, s_{f_1(a)}, x)^-$.
\item[Step 2:] $(s_{f_1(r)}^*, x, s_{f_0(a+1)}, s_{f_0(0)})^+$ and $(s_{f_0(0)}^*, s_{f_0(a+1)}^*, s_{f_3(a+1)}, x)^-$.
\end{enumerate}
are used.
This set of monomers results in the sequence of insertions 
$$ (s_{f_1(a)}, s_{f_1(r)}) \diamond (s_{f_1(r)}^*, s_{f_1(a)}^*) $$
$$ (s_{f_1(a)}, s_{f_1(r)}) \bm{(s_{f_1(r)}^*, x, s_{f_3(a+1)}, s_{f_1(r)})} \diamond (s_{f_1(r)}^*, s_{f_1(a)}^*) $$
$$ (s_{f_3(a+1)}, s_{f_1(r)}) \diamond (s_{f_1(r)}^*, s_{f_1(a)}^*) $$
$$ (s_{f_3(a+1)}, s_{f_1(r)}) \diamond \bm{(s_{f_0(0)}^*, s_{f_3(a+1)}^*, s_{f_1(a)}, x)} (s_{f_1(r)}^*, s_{f_1(a)}^*) $$
$$ (s_{f_3(a+1)}, s_{f_1(r)}) \diamond (s_{f_0(0)}^*, s_{f_3(a+1)}^*) $$
$$ (s_{f_3(a+1)}, s_{f_1(r)}) \bm{(s_{f_1(r)}^*, x, s_{f_0(a+1)}, s_{f_0(0)})} \diamond (s_{f_0(0)}^*, s_{f_3(a+1)}^*) $$
$$ (s_{f_0(a+1)}, s_{f_0(0)}) \diamond (s_{f_0(0)}^*, s_{f_3(a+1)}^*) $$
$$ (s_{f_0(a+1)}, s_{f_0(0)}) \diamond \bm{(s_{f_0(0)}^*, s_{f_0(a+1)}^*, s_{f_3(a+1)}, x)} (s_{f_0(0)}^*, s_{f_3(a+1)}^*) $$
$$ (s_{f_0(a+1)}, s_{f_0(0)}) \diamond (s_{f_0(0)}^*, s_{f_0(a+1)}^*) $$

The first insertion site $(s_{f_1(a)}, s_{f_1(r)}) (s_{f_1(r)}^*, s_{f_1(a)}^*)$ accepts any monomer of the form $(\underline{~~}, s_{f_1(a)}^*, s_{f_1(a)}, \underline{~~})^-$ or $(s_{f_1(r)}^*, \underline{~~}, \underline{~~}, s_{f_1(r)})^+$.
There are two such monomer types in the system: $(s_{f_0(i+1)}^*, s_{f_1(a)}^*, s_{f_1(a)}, x)^-$ where $i$ is an odd integer with $0 \leq i < r$ (inner odd-to-even monomers) and $(s_{f_1(r)}^*, x, s_{f_3(i+1)}, s_{f_1(r)})^+$ with $0 \leq i < r$ (outer monomers).
Insertion of the first type has the form
$$ (s_{f_1(a)}, s_{f_1(r)}) \diamond (s_{f_1(r)}^*, s_{f_1(a)}^*) $$ 
$$ (s_{f_1(a)}, s_{f_1(r)}) (s_{f_0(i+1)}^*, s_{f_1(a)}^*, s_{f_1(a)}, x) (s_{f_1(r)}^*, s_{f_1(a)}^*) $$ 

Neither resulting site is insertable, since the left site requires a monomer of the form $(s_{f_1(r)}^*, \underline{~~}, \underline{~~}, s_{f_0(i+1)})^+$ not found in the system.
Insertion of the second type has the form
$$ (s_{f_1(a)}, s_{f_1(r)}) \diamond (s_{f_1(r)}^*, s_{f_1(a)}^*) $$ 
$$ (s_{f_1(a)}, s_{f_1(r)}) (s_{f_1(r)}^*, x, s_{f_3(i+1)}, s_{f_1(r)}) (s_{f_1(r)}^*, s_{f_1(a)}^*) $$ 

Only the right resulting site is insertable.
Moreover, this site is only insertable if $i = a$, since the only monomer type of the form $(\underline{~~}, s_{f_3(i+1)}^*, s_{f_1(a)}, \underline{~~})^-$ has $i = a$. 
So only one monomer type inserted into the first site enables further growth: $(s_{f_1(r)}, x, s_{f_3(a+1)}, s_{f_1(r)})^+$.
So the first insertion site is growth-deterministic.

The second insertion site $(s_{f_3(a+1)}, s_{f_1(r)})(s_{f_1(r)}^*, s_{f_1(a)}^*)$ accepts any monomer of the form $(\underline{~~}, s_{f_3(a+1)}^*, s_{f_1(a)}, \underline{~~})^-$.
The only such monomer type in the system is $(s_{f_0(0)}^*, s_{f_3(a+1)}^*, s_{f_1(a)}, x)^-$.
So the second insertion site is deterministic and thus growth-deterministic.

The third insertion site $(s_{f_3(a+1)}, s_{f_1(r)})(s_{f_0(0)}^*, s_{f_3(a+1)}^*)$ accepts any monomer of the form $(s_{f_1(r)}^*, \underline{~~}, \underline{~~}, s_{f_0(0)})^+$.
The only such monomer types in the system are $(s_{f_1(r)}, x, s_{f_0(i+1)}, s_{f_0(0)})^+$ with $0 \leq i < r$.
So the insertion has the form
$$ (s_{f_3(a+1)}, s_{f_1(r)}) \diamond (s_{f_0(0)}^*, s_{f_3(a+1)}^*) $$ 
$$ (s_{f_3(a+1)}, s_{f_1(r)}) (s_{f_1(r)}, x, s_{f_0(i+1)}, s_{f_0(0)}) (s_{f_0(0)}^*, s_{f_3(a+1)}^*) $$ 

Only the resulting right site is insertable.
Moreover, this site is only insertable if $i = a$, since the only monomers of the form $(\underline{~~}, s_{f_0(i+1)}, s_{f_3(a+1)}, \underline{~~})^-$ in the system have $i=a$. 
So only one monomer type inserted into the third site enables further growth: $(s_{f_1(r)}, x, s_{f_0(a+1)}, s_{f_0(0)})^+$.
So the third insertion site is growth-deterministic. 

The fourth insertion site $(s_{f_0(a+1)}, s_{f_0(0)})(s_{f_0(0)}^*, s_{f_3(a+1)}^*)$ accepts any monomer of the form $(\underline{~~}, s_{f_0(a+1)}^*, s_{f_3(a+1)}, \underline{~~})^-$.
The only such monomer types in the system are $(s_{f_0(0)}^*, s_{f_0(i+1)}^*, s_{f_3(i+1)}, x)^-$ with $0 \leq i < r$.
So the insertion has the form
$$ (s_{f_0(a+1)}, s_{f_0(0)}) \diamond (s_{f_0(0)}^*, s_{f_3(a+1)}^*) $$
$$ (s_{f_0(a+1)}, s_{f_0(0)}) (s_{f_0(0)}^*, s_{f_0(i+1)}^*, s_{f_3(i+1)}, x) (s_{f_0(0)}^*, s_{f_3(a+1)}^*) $$

Only the resulting left site is insertable.
Moreover, this site is only insertable if $i = a$, since the only monomers of the form $(\underline{~~}, s_{f_0(a+1)}^*, s_{f_0(i+1)}, \underline{~~})^-$ or $(s_{f_0(0)}^*, \underline{~~}, \underline{~~}, s_{f_0(0)})^+$ have $i = a$.
So only one monomer type inserted into the fourth site enables further growth: $(s_{f_0(0)}^*, s_{f_0(i+1)}^*, s_{f_3(i+1)}, x)^-$.
So the fourth insertion site is growth-deterministic.

\textbf{Inner even-to-odd increment.}
We implement inner even-to-odd increments in three steps.
The first step is a sequence of three insertions:
$$ (s_{f_0(a)}, s_{f_0(b)}) \diamond (s_{f_0(b)}^*, s_{f_0(a)}^*) $$ 
$$ (s_{f_0(a)}, s_{f_0(b)}) \diamond \bm{(s_{f_1(b+1)}^*, s_{f_0(a)}^*, s_{f_0(a)}, x)} (s_{f_0(b)}^*, s_{f_0(a)}^*) $$ 
$$ (s_{f_0(a)}, s_{f_0(b)}) \diamond (s_{f_1(b+1)}^*, s_{f_0(a)}^*) $$
$$ (s_{f_0(a)}, s_{f_0(b)}) \bm{(s_{f_0(b)}^*, x, s_{f_1(a)}, s_{f_1(b+1)})} \diamond (s_{f_1(b+1)}^*, s_{f_0(a)}^*) $$
$$ (s_{f_1(a)}, s_{f_1(b+1)}) \diamond (s_{f_1(b+1)}^*, s_{f_0(a)}^*) $$
$$ (s_{f_1(a)}, s_{f_1(b+1)}) \diamond \bm{(s_{f_1(b+1)}^*, s_{f_1(a)}^*, s_{f_0(a)}, s_{f_2(b+2)})} \diamond (s_{f_1(b+1)}^*, s_{f_0(a)}^*) $$

This results in two insertion sites.
The left insertion site is an incremented version of the original site, while the right site is not.
The second step is used to modify the right site into an incremented version of the original site, but incremented \emph{by~2}:
$$ (s_{f_0(a)}, s_{f_2(b+2)}) \diamond (s_{f_1(b+1)}^*, s_{f_0(a)}^*) $$
$$ (s_{f_0(a)}, s_{f_2(b+2)}) \diamond \bm{(s_{f_2(b+2)}^*, s_{f_2(a)}^*, x, s_{f_1(b+1)})}  (s_{f_1(b+1)}^*, s_{f_0(a)}^*) $$
$$ (s_{f_0(a)}, s_{f_2(b+2)}) \diamond (s_{f_2(b+2)}^*, s_{f_2(a)}^*) $$
$$ (s_{f_0(a)}, s_{f_2(b+2)}) \diamond \bm{(s_{f_0(b+2)}^*, s_{f_0(a)}^*, s_{f_2(a)}, x)} (s_{f_2(b+2)}^*, s_{f_2(a)}^*) $$
$$ (s_{f_0(a)}, s_{f_2(b+2)}) \diamond (s_{f_0(b+2)}^*, s_{f_0(a)}^*) $$
$$ (s_{f_0(a)}, s_{f_2(b+2)}) \bm{(s_{f_2(b+2)}^*, x, s_{f_3(a)}, s_{f_0(b+2)})} \diamond (s_{f_0(b+2)}^*, s_{f_0(a)}^*) $$
$$ (s_{f_3(a)}, s_{f_0(b+2)}) \diamond (s_{f_0(b+2)}^*, s_{f_0(a)}^*) $$
$$ (s_{f_3(a)}, s_{f_0(b+2)}) \diamond \bm{(x, s_{f_3(a)}^*, s_{f_0(a)}, s_{f_0(b+2)})} (s_{f_0(b+2)}^*, s_{f_0(a)}^*) $$
$$ (s_{f_0(a)}, s_{f_0(b+2)}) (s_{f_0(b+2)}^*, s_{f_0(a)}^*) $$

In the case that $b+2 > r$, the site is simply uninsertable.

The first site $(s_{f_0(a)}, s_{f_0(b)})(s_{f_0(b)}^*, s_{f_0(a)}^*)$ accepts any monomer of the form $(s_{f_0(b)}^*, \underline{~~}, \underline{~~}, s_{f_0(b)})^+$ or $(\underline{~~}, s_{f_0(a)}^*, s_{f_0(a)}, \underline{~~})^-$.
The only such monomer types in the system are $(s_{f_1(i+1)}^*, s_{f_0(a)}^*, s_{f_0(a)}, x)^-$ with $0 \leq i < r$.
So the insertion has the form
$$ (s_{f_0(a)}, s_{f_0(b)}) \diamond (s_{f_0(b)}^*, s_{f_0(a)}^*) $$
$$ (s_{f_0(a)}, s_{f_0(b)}) (s_{f_1(i+1)}^*, s_{f_0(a)}^*, s_{f_0(a)}, x) (s_{f_0(b)}^*, s_{f_0(a)}^*) $$

Only the resulting left site is insertable.
Moreover, this site is only insertable if $i = b$, since the only monomers of the form $(s_{f_0(b)}^*, \underline{~~}, \underline{~~}, s_{f_1(i+1)})^+$ have $i = b$. 
So only one monomer type inserted into the first site enables further growth: $(s_{f_1(b+1)}^*, s_{f_0(a)}^*, s_{f_0(a)}, x)^+$.
So the first insertion site is growth-deterministic.

The second site $(s_{f_0(a)}, s_{f_0(b)})(s_{f_1(b+1)}^*, s_{f_0(a)}^*)$ accepts any monomer of the form $(s_{f_0(b)}^*, \underline{~~}, \underline{~~}, s_{f_1(b+1)})^+$.
The only such monomer types in the system are $(s_{f_0(b)}, x, s_{f_1(i)}, s_{f_1(b+1)})^+$ with $0 \leq i \leq r$.
So the insertion has the form
$$ (s_{f_0(a)}, s_{f_0(b)}) \diamond (s_{f_1(b+1)}^*, s_{f_0(a)}^*) $$
$$ (s_{f_0(a)}, s_{f_0(b)}) (s_{f_0(b)}, x, s_{f_1(i)}, s_{f_1(b+1)}) (s_{f_1(b+1)}^*, s_{f_0(a)}^*) $$

Only the right resulting site is insertable.
Moreover, this site is only insertable if $i = a$, since the only monomers of the form $(\underline{~~}, s_{f_1(i)}^*, s_{f_0(a)}, \underline{~~})^-$ have $i = a$. 
So only one monomer type inserted into the second site enables further growth: $(s_{f_0(b)}, x, s_{f_1(a)}, s_{f_1(b+1)})^+$.
So the second site is growth-deterministic. 

The third site $(s_{f_1(a)}, s_{f_1(b+1)})(s_{f_1(b+1)}^*, s_{f_0(a)}^*)$ accepts any monomer of the form $(\underline{~~}, s_{f_1(a)}^*, s_{f_0(a)}, \underline{~~})^-$.
The only such monomer types in the system are $(s_{f_1(i+1)}^*, s_{f_1(a)}^*, s_{f_0(a)}, s_{f_2(i+2)})^-$ with $0 \leq i < r$.
So the insertion has the form
$$ (s_{f_1(a)}, s_{f_1(b+1)}) \diamond (s_{f_1(b+1)}^*, s_{f_0(a)}^*) $$
$$ (s_{f_1(a)}, s_{f_1(b+1)}) (s_{f_1(i+1)}^*, s_{f_1(a)}^*, s_{f_0(a)}, s_{f_2(i+2)}) (s_{f_1(b+1)}^*, s_{f_0(a)}^*) $$

Both of the resulting sites are insertable.
The left site is only insertable if $i = b$, since the only monomers of the form $(s_{f_1(b+1)}^*, \underline{~~}, \underline{~~}, s_{f_1(i+1)})^+$ or $(\underline{~~}, s_{f_1(a)}^*, s_{f_1(a)}, \underline{~~})^-$ require $i = b$ to insert.
So only one monomer type inserted into the third site enables further growth: $(s_{f_1(b+1)}^*, s_{f_1(a)}^*, s_{f_0(a)}, s_{f_2(b+2)})^-$.

The right site is also only insertable if $i = b$, since the only monomers of the form $(s_{f_2(i+2)}^*, \underline{~~}, \underline{~~}, s_{f_1(b+1)})^+$ have $i = b$.
So only one monomer type inserted into the third site enables further growth: $(s_{f_1(b+1)}^*, s_{f_1(a)}^*, s_{f_0(a)}, s_{f_2(b+2)})^-$.
So the third site is growth-deterministic.

The fourth site $(s_{f_0(a)}, s_{f_2(b+2)})(s_{f_1(b+1)}^*, s_{f_0(a)}^*)$ accepts any monomer of the form $(s_{f_2(b+2)}^*, \underline{~~}, \underline{~~}, s_{f_1(b+1)})^+$.
The only such monomer types in the system are $(s_{f_2(b+2)}^*, s_{f_2(i)}^*, x, s_{f_1(b+1)})^+$ with $0 \leq i \leq r$.
So the insertion has the form
$$ (s_{f_0(a)}, s_{f_2(b+2)}) \diamond (s_{f_1(b+1)}^*, s_{f_0(a)}^*)$$
$$ (s_{f_0(a)}, s_{f_2(b+2)}) (s_{f_2(b+2)}^*, s_{f_2(i)}^*, x, s_{f_1(b+1)}) (s_{f_1(b+1)}^*, s_{f_0(a)}^*)$$

Only the left resulting site is insertable.
Moreover, this site is only insertable if $i = a$, since the only monomers of the form $(\underline{~~}, s_{f_0(a)}^*, s_{f_2(i)}, \underline{~~})^-$ have $i = a$.
So only one monomer type inserted into the fourth site enables further growth: $(s_{f_2(b+2)}^*, s_{f_2(a)}, x, s_{f_1(b+1)})^+$.
So the fourth site is growth-deterministic.

The fifth site $(s_{f_0(a)}, s_{f_2(b+2)})(s_{f_2(b+2)}^*, s_{f_2(a)}^*)$ accepts any monomer of the form $(\underline{~~}, s_{f_0(a)}^*, s_{f_2(a)}, \underline{~~})^-$.
The only such monomer types in the system are $(s_{f_0(i+2)}^*, s_{f_0(a)}, s_{f_2(a)}, x)^-$ with $0 \leq i < r$.
So the insertion has the form 
$$ (s_{f_0(a)}, s_{f_2(b+2)}) \diamond (s_{f_2(b+2)}^*, s_{f_2(a)}^*) $$
$$ (s_{f_0(a)}, s_{f_2(b+2)}) (s_{f_0(i+2)}^*, s_{f_0(a)}, s_{f_2(a)}, x) (s_{f_2(b+2)}^*, s_{f_2(a)}^*) $$

Only the left resulting site is insertable.
Moreover, this site is only insertable if $i = b$, since the only monomers of the form $(s_{f_2(b+2)}^*, \underline{~~}, \underline{~~}, s_{f_0(i+2)})^+$ have $i = b$.
So only one monomer type inserted into the fifth site enables further growth: $(s_{f_0(b+2)}^*, s_{f_0(a)}, s_{f_2(a)}, x)^-$.
So the fifth site is growth-deterministic.

The sixth site $(s_{f_0(a)}, s_{f_2(b+2)})(s_{f_0(b+2)}^*, s_{f_0(a)}^*)$ accepts any monomer type of the form $(s_{f_2(b+2)}^*, \underline{~~}, \underline{~~}, s_{f_0(b+2)})^-$.
The only such monomer types in the system are $(s_{f_2(b+2)}^*, x, s_{f_3(i)}, s_{f_0(b+2)})^+$ with $0 \leq i \leq r$.
So the insertion has the form
$$(s_{f_0(a)}, s_{f_2(b+2)}) \diamond (s_{f_0(b+2)}^*, s_{f_0(a)}^*)$$ 
$$(s_{f_0(a)}, s_{f_2(b+2)}) (s_{f_2(b+2)}^*, x, s_{f_3(i)}, s_{f_0(b+2)}) (s_{f_0(b+2)}^*, s_{f_0(a)}^*)$$ 

Only the right resulting site is insertable.
Moreover, this site is only insertable if $i = a$, since the only monomers of the form $(\underline{~~}, s_{f_3(i)}, s_{f_0(a)}, \underline{~~})^-$ have $i = a$.
So only one monomer type inserted into the sixth site enables further growth: $(x, s_{f_3(a)}, s_{f_0(a)}, s_{f_0(b+2)})^-$.
So the sixth site is growth-deterministic.

The seventh site $(s_{f_3(a)}, s_{f_0(b+2)})(s_{f_0(b+2)}^*, s_{f_0(a)}^*)$ accepts any monomer type of the form $(\underline{~~}, s_{f_3(a)}^*, s_{f_0(a)}, \underline{~~})^-$.
The only such monomer types in the system are $(x, s_{f_3(a)}^*, s_{f_0(a)}, s_{f_0(i+2)})^-$ with $0 \leq i < r$.
So the insertion has the form
$$ (s_{f_3(a)}, s_{f_0(b+2)}) \diamond (s_{f_0(b+2)}^*, s_{f_0(a)}^*) $$
$$ (s_{f_3(a)}, s_{f_0(b+2)}) (x, s_{f_3(a)}^*, s_{f_0(a)}, s_{f_0(i+2)}) (s_{f_0(b+2)}^*, s_{f_0(a)}^*) $$
 
Only the right resulting site is insertable.
Moreover, this site is only insertable if $i = b$, since the only monomers of the form $(s_{f_0(i+2)}, \underline{~~}, \underline{~~}, s_{f_0(b+2)})^+$ or $(\underline{~~}, s_{f_0(a)}^*, s_{f_0(a)}, \underline{~~})^-$ require $i = b$ to insert.
So only one monomer type inserted into the seventh site enables further growth: $(x, s_{f_3(a)}^*, s_{f_0(a)}, s_{f_0(b+2)})^-$.
So the seventh site is growth-deterministic.

\textbf{Putting it together.}
Since all insertion sites are growth-deterministic, and the correct monomer types carry out the execution of a double for-loop, the system constructs only polymers of length at most the length of the polymer where all double for-loops are executed to completion.
The size $P_i$ of a subpolymer with initiator encoding some value between $0$ and $(r+1)^2-1$ can be bounded by $2P_{i+2} + 7 \leq P_i \leq 2P_{i+1} + 7$, since either $i+1$ or $i+2$ is an inner even-to-odd increment step (inserting~7 monomers) and no step inserts more than~7 monomers.
Moreover, $P_{(r+1)^2-2} > 0$.
So $P_0 + 2$, the size of the longest terminal polymer of the system, is $2^{\Theta(r^2)}$.

\textbf{Running time.}
Recall that the purpose of this construction was to achieve reduced running time by allowing multiple monomer types to ``compete'' for common sites.
Each monomer form listed in Table~\ref{tab:all-monomers-speed-extreme-ub} corresponds to a 1, $r$, or $r+1$ sets of monomer types, with each set consisting of types insertable into common sites.
For instance, the form $(s_{f_1(b+1)}^*, s_{f_0(a)}^*, s_{f_0(a)}, x)^-$ corresponds to $r+1$ sets of monomer types, one for each value of $a$, where all types in a set can be inserted into exactly the sites of the form $(s_{f_0(a)}, c) (\overline{c}, s_{f_0(a)}^*)$.

We assign equal concentrations to each set of monomer types, with the concentrations of the types within the set distributed equally.
There are $2 + 6r + 6(r+1) \leq 2r + 6r + 12r = 20r$ total sets, so the total concentration of the types in each set is at least $1/(20r)$.
We consider the expected construction time of the longest polymer of the system.
The construction of this polymer is complete as soon as every counter's variables have reached the value $a = b = r$, i.e. the insertion site has been modified to be $(s_{f_1(r)}, s_{f_1(r)}) (s_{f_1(r)}^*, s_{f_1(r)}^*)$ and the monomer $(s_{f_1(r)}^*, x, s_{f_3(a+1)}, s_{f_1(r)})^+$ has been inserted.
There are fewer than $2^{r^2}$ such insertions, and each insertion can occur once at most $7 \cdot (r+1)^2 \leq 28r^2$ previous insertions have occurred.

So an upper bound on the expected time $T_r$ for each such insertion is described as a sum of $28r^2$ random variables, each with expected time $20r$.
The Chernoff bound for exponential random variables implies the following upper bound on $T_r$:

$\begin{aligned}
{\rm Prob}[T_r > 20r \cdot 28r^2(1 + \delta)] &\leq e^{-20 \cdot 28r^3 \delta^2 / (2 + \delta)} \\
&\leq e^{-r^3 \delta^2 / (2 + \delta)} \\
&\leq e^{-r^3 \delta^2 / (2\delta)} \rm{~for~all~} \delta \geq 2 \\
&\leq e^{-r^3 \delta / 2} \\
\end{aligned}$

Let $T_{\mathcal{S}_r}$ be the total running time of the system. 
Then we can bound $T_{\mathcal{S}_r}$ from above using the bound for $T_r$:

$\begin{aligned}
{\rm Prob}[T_{\mathcal{S}_r} > 20r \cdot 28r^2(1 + \delta)] &\leq 2^{r^2} \cdot e^{-r^3 \delta / 2} \\
&\leq 2^{r^2} 2^{-r^3 \delta/2} \\ 
&\leq 2^{r^2 - r^3 \delta/2} \\
&\leq 2^{r^3 \delta/4 - r^3 \delta/2} \text{~for~all~} \delta \geq 4 \\
&\leq 2^{-r^3\delta/4}
\end{aligned}$

So ${\rm Prob}[T_{\mathcal{S}_r} > 20r \cdot 28r^2(1 + \delta)] \leq 2^{-r^3\delta/4}$ for all $\delta \geq 4$.
So the expected value of $T_{\mathcal{S}_r}$, the construction time, is $O(r^3) = O(\log^{3/2}(n))$ with an exponentially decaying tail probability.
\qed
\end{proof}

\section{Negative Results for Polymer Growth}
\label{sec:negative-results}

Here we show that the constructions in the previous section are asymptotically the best possible.
This is done in Theorems~\ref{thm:deterministic-lb},~\ref{thm:types-extreme-lb} and~\ref{thm:speed-extreme-lb}.
A collection of intervening lemmas are used to prove bounds on the number of monomer types and expected time to carry out an \emph{insertion sequence}: a sequence of monomer insertions where each insertion is into a site created by the previous insertion.

Observe that if two monomer types of the same sign are insertable into a common site, then the set of sites each can be inserted into is equal.
Nearly all of the lemmas involve consideration of not only monomer types, but \emph{insertion sets}: maximal sets of same-signed monomer types sharing a common set of insertion sites each can be inserted into. 
The system described in Theorem~\ref{thm:types-extreme-ub} uses $\Theta(\log^{2/3}(n))$ singleton insertion sets.
On the other hand, the construction of Theorem~\ref{thm:speed-extreme-ub} uses $\Theta(\sqrt{\log{n}})$ insertion sets, each of $\Theta(\sqrt{\log{n}})$ monomer types, e.g. $\{(s_{f_1(b+1)}^*, s_{f_0(a)}^*, s_{f_0(a)}, x)^- : 0 \leq b < r \}$ in the first row of Table~\ref{tab:all-monomers-speed-extreme-ub}.
These larger sets decrease construction time by increasing the total concentration of monomer types insertable into each site.

The first several lemmas of the section are used to prove Theorem~\ref{thm:deterministic-lb} and Lemma~\ref{lem:trade-off-lb}, a lemma describing the trade-off between the number of monomer types and expected construction time for systems constructing finite polymer sets.
This lemma is combined with extremal bounds on the minimum number of monomer types and insertion sets to prove that the constructions of Theorem~\ref{thm:types-extreme-ub} and~\ref{thm:speed-extreme-ub} are optimal at both ends of this trade-off.
 
\begin{lemma}
\label{lem:positive-only-sites}
Any insertion sequence of length $l$ with no repeated insertion sites has $\Theta(l)$ sites of the form $(a, b) (\overline{c}, \overline{a})$ with $b \neq c$.
\end{lemma}

\begin{proof}
Insertion sites have one of three forms: 

\begin{enumerate}[leftmargin=2cm] \itemsep5pt
\item[Positive:] $(a, b) (\overline{c}, \overline{a})$ with $b \neq c$.
\item[Mixed:] $(a, b) (\overline{b}, \overline{a})$.
\item[Negative:] $(b, a) (\overline{a}, \overline{c})$ with $b \neq c$.
\end{enumerate}

We prove that every sequence of four consecutive insertion sites has at least one positive site.
Consider such a sequence of four sites (and the three intervening insertions).
If the first site is positive, we're done.
If the first site is mixed, then the first monomer type inserted may be negative or positive.
If a negative monomer type is inserted:

$$ (a, b) \diamond (\overline{b}, \overline{a}) $$ 
$$ (a, b) (\overline{c}, \overline{a}, a, d) (\overline{b}, \overline{a}) $$ 

Since the sequence does not repeat sites, $b \neq c, d$ and either second site, $(a, b)(\overline{c}, \overline{a})$ or $(a, d) (\overline{b}, \overline{a})$, is positive.
If a positive monomer type is inserted:

$$ (a, b) \diamond (\overline{b}, \overline{a}) $$ 
$$ (a, b) (\overline{b}, \overline{c}, d, b) (\overline{b}, \overline{a}) $$ 

As before, $a, d \neq c$ since sites cannot repeat. 
So the next insertion must use a negative monomer type.
We assume the left site is used next in the sequence (a symmetric argument works if the right site is used instead).
The entire insertion sequence has the form: 

$$ (a, b) \diamond (\overline{b}, \overline{a}) $$ 
$$ (a, b) \diamond (\overline{b}, \overline{c}, d, b) (\overline{b}, \overline{a}) $$ 
$$ (a, b) \diamond (\overline{b}, \overline{c}) $$
$$ (a, b) (\overline{e}, \overline{a}, c, f) (\overline{b}, \overline{c}) $$

As before, $e, f \neq b$ since sites cannot repeat.
So the next site, either $(a, b) (\overline{e}, \overline{a})$ with $b \neq e$ or $(c, f) (\overline{b}, \overline{c})$ with $f \neq b$ is positive.
So the third site in the sequence is positive.
Finally, if the intial site is negative then the first monomer type inserted is negative:

$$ (b, a) \diamond (\overline{a}, \overline{c}) $$ 
$$ (b, a) \diamond (\overline{d}, \overline{b}, c, e) (\overline{a}, \overline{c}) $$ 
$$ (b, a) \diamond (\overline{d}, \overline{b}) $$

We assume that the left site is used next in the sequence (a symmetric argument works if the right side is used instead).
If $a \neq d$ then the second site is positive.
Otherwise the second site is mixed, and by previous argument, at most two more insertions (a total of three) will take place until a positive site appears.
So in the entire sequence of length $l$, a positive site appears at least once in every sequence of four consecutive sites.
\qed
\end{proof}

\begin{lemma}
\label{lem:usable-sites-ub}
Any insertion sequence with no repeated insertion sites using $k$ monomer types forming $m$ insertion sets has length $O(m\sqrt{k})$.
\end{lemma}

\begin{proof}
Let $\mathcal{S} = (\Sigma, \Delta, Q, R)$ be the insertion system containing the sequence.
Relabel the symbols in $\Sigma \cup \{s^* : s \in \Sigma\}$ as $s_1, s_2, \dots, s_{4k}$, with some of these symbols possibly unused.
Note that this implies that for every $s_i$, $\overline{s_i} = s_j$ for some $j \in 1, 2 \dots, 4k$.
Let $l$ be the length of the sequence.
By Lemma~\ref{lem:positive-only-sites}, $\Theta(l)$ sites are \emph{positive}: they have the form $(s_a, s_b) (\overline{s_c}, \overline{s_a})$ with $b \neq c$.

\textbf{A bound of $\sum_{i=1}^{4k}{\rm{min}(|L_i|, |R_i|)} \leq 3m$.}
Let $L_i$ and $R_i$ be the sets of monomer types of the forms $(\underline{~~}, \underline{~~}, s_i, \underline{~~})^{\pm}$ and $(\underline{~~}, \overline{s_i}, \underline{~~}, \underline{~~})^{\pm}$, respectively, used in the insertion sequence.
Each positive site $(s_i, s_b) (\overline{s_c}, \overline{s_i})$ consists of a left monomer in $L_i$ and right monomer in $R_i$.
Every occurrence of a positive site in the sequence is followed by the use of the left or right resulting site, e.g.:

$$ (s_i, s_b) \diamond (\overline{s_c}, \overline{s_i}) $$
$$ (s_i, s_b) \diamond (\overline{s_b}, \overline{s_d}, s_e, s_c) (\overline{s_c}, \overline{s_i}) $$
$$ (s_i, s_b) \diamond (\overline{s_b}, \overline{s_d}) $$

It is the case that $d$ is unique for $c$, i.e. no two insertions into positive sites using the left resulting sites both use monomers of the form $(\overline{s_b}, \overline{s_d}, \underline{~~}, \underline{~~})^+$, since such a pair of monomers implies the sequence repeats the insertion site $(s_a, s_b) (\overline{s_b}, \overline{s_d})$.
A similar claim holds for $e$ and $b$ in the case that the right site is used.
So inserting into the resulting site requires a monomer from a distinct insertion set $\{(\underline{~~}, \overline{s_i}, s_d, \underline{~~})^- \in \Delta \}$ or, in the special case that $i = d$, $\{(\overline{s_b}, \underline{~~}, \underline{~~}, s_b)^+ \in \Delta \}$.

The resulting sites require monomers from a number of distinct insertion sets equal to the sum of two values.
First, the number of times the left side is used with a distinct $c$ and a monomer is inserted into a site $(s_i, s_b)(\overline{s_b}, \overline{s_d})$ with $d$ unique for $c$.
Second, the number of times the right side is used with a distinct $b$ and a monomer is inserted into a site $(s_e, s_c)(\overline{s_c}, \overline{s_i})$ with $e$ unique for $b$.
An assignment of left and right side usage that minimizes the number of distinct insertion sets needed is nearly equivalent to a minimum vertex covering of the following bipartite graph:

\begin{itemize}[label=\textbullet] \itemsep5pt
\item A node $L_{(i, b)}$ for every site $(s_i, s_b)(\overline{s_c}, \overline{s_i})$ in the insertion sequence.
\item A node $R_{(c, i)}$ for every site $(s_i, s_b)(\overline{s_c}, \overline{s_i})$ in the insertion sequence.
\item An edge $(L_{(i, b)}, R_{(c, i)})$ for every site $(s_i, s_b) (\overline{s_c}, \overline{s_i})$ in the insertion sequence. 
\end{itemize}

Selecting a vertex to cover a given edge corresponds to using the resulting left or right site of the edge's site, e.g. selecting $R_{(c, i)}$ to cover the edge $(L_{(i, b)}, R_{(c, i)})$ corresponds to using the resulting left site and inserting a monomer type of the form $(\underline{~~}, \overline{s_i}, s_d, \underline{~~})^-$, where $d$ is unique for $c$.
 By K\"{o}nig's theorem (see~\cite{Diestel-2005a,Bondy-1976a}), since the graph is bipartite, the size of a minimum vertex covering is equal to the size of a maximum matching, which is bounded from above by $\sum_{i=1}^{4k}\rm{min}(|L_i|, |R_i|)$.

However, an insertion set $\{(\underline{~~}, \overline{s_e}, s_d, \underline{~~})^- \in \Delta\}$ corresponds to selecting both $R_{(c, i)}$, where $d$ is unique for $c$, and $L_{(j, b)}$, where $e$ is unique for $b$. 
So the number of insertion sets needed may be as little as half the size of the vertex cover of the bipartite graph.
Additionally, one site may not be inserted into.
So $\sum_{i=1}^{4k}\rm{min}(|L_i|, |R_i|) - 1 \leq 2m$ and $\sum_{i=1}^{4k}\rm{min}(|L_i|, |R_i|) \leq 3m$.

\textbf{Maximizing insertion sequence length.}
Consider the number of positive sites $y$ accepting some monomer type.
We proved that $\Omega(l) = y $ and it is easily observed that $y \leq \sum_{i=1}^{4k}{\rm{min}(m, |L_i|\cdot|R_i|)}$.
We also proved that $\sum_{i=1}^{4k}{\rm{min}(|L_i|, |R_i|)} \leq 3m$ and it is easily observed that $\sum_{i=1}^{4k}{\rm{max}(|L_i|, |R_i|)} \leq 2k$, since each monomer type is in at most one $L_i$ and one $R_i$.
This gives the following set of constraints:

\begin{enumerate} \itemsep5pt
\item $\Omega(l) = \sum_{i=1}^{4k}{\rm{min}(m, |L_i|\cdot|R_i|)}$.
\item $\sum_{i=1}^{4k}{\rm{min}(|L_i|, |R_i|)} \leq 3m$.
\item $\sum_{i=1}^{4k}{\rm{max}(|L_i|, |R_i|)} \leq 2k$.
\end{enumerate}

Observe that $|L_i|\cdot|R_i| = \rm{min}(|L_i|, |R_i|) \cdot \rm{max}(|L_i|, |R_i|)$.
Define two new variables $y_i = \rm{min}(|L_i|, |R_i|)$ and $z_i = \rm{max}(|L_i|, |R_i|)$ for an alternate formulation of the previous constraints:

\begin{enumerate} \itemsep5pt
\item $\Omega(l) = \sum_{i=1}^{4k}{\rm{min}(m, y_iz_i)}$.
\item $\sum_{i=1}^{4k}{y_i} \leq 3m$.
\item $\sum_{i=1}^{4k}{z_i} \leq 2k$.
\end{enumerate}

Relax $y_i$, $z_i$ to be real-valued and let $W = \{ i : y_iz_i > 0\}$. 
If $0 < y_iz_i, y_jz_j < m$ for some $i \neq j$ and $y_i = \rm{max}(y_i, z_i, y_j, z_j)$, then $\rm{min}(m, y_iz_i) + \rm{min}(m, y_jz_j) < \rm{min}(m, y_i(z_i + \varepsilon)) + \rm{min}(m, y_j(z_j - \varepsilon))$ for sufficiently small $\varepsilon > 0$.
More generally, if $0 < y_iz_i, y_jz_j < m$ for some $i \neq j$ then the values of $y_i, z_i, y_j, z_j$ can be modified to increase $\sum_{i=1}^{4k}{\rm{min}(m, y_iz_i)}$.
Therefore the maximum value is achieved when $m = y_iz_i$ for all but at most one $i \in W$.

We claim that it cannot be that $y_iz_i = m$ for $6\sqrt{k}$ distinct values of $i$.
By contradiction, assume so.
So $|W| \geq 6\sqrt{k}$ and the average value of $y_i$ for $i \in W$ must be less than $3m/(6\sqrt{k}) = m/(2\sqrt{k})$.
So for a subset $W' \subseteq W$ with $|W'| \geq |W|/2 \geq 3\sqrt{k}$, $y_i \leq 2 \cdot m/(2\sqrt{k}) = m/\sqrt{k}$ for all $i \in W'$.
For every $i \in W'$, because $y_i \leq m/\sqrt{k}$ and $y_iz_i = m$, it must be the case that $z_i \geq \sqrt{k}$.
So $\sum_{i=1}^{4k}{z_i} \geq |W'| \cdot \sqrt{k} \geq 3k$, a contradiction with the constraint that $\sum_{i=1}^{4k}{z_i} \leq 2k$.

So the maximum value is achieved when $m = y_iz_i$ for all but at most one $i \in W$, with $|W| + 1 < 6\sqrt{k} + 1 < 7\sqrt{k}$.
So $\sum_{i=1}^{4k}{\rm{min}(m, y_iz_i)} \leq (|W|+1)m < 7m\sqrt{k}$.
So $\Omega(l) = 7m\sqrt{k}$ and $l = O(m\sqrt{k})$. 
\qed
\end{proof}

\begin{lemma}
\label{lem:speed-lb}
An insertion sequence of length $l$ using monomer types from $m$ insertion sets with no repeated insertion sites takes $\Omega(ml)$ expected time.
\end{lemma}

\begin{proof}
By linearity of expectation, the total expected time of the insertions is equal to the sum of the expected time for each insertion.
By Lemma~\ref{lem:positive-only-sites}, $\Theta(l)$ sites are both \emph{positive}, i.e.\ they have the form $(s_a, s_b) (\overline{s_c}, \overline{s_a})$ with $b \neq c$, and accept the monomer types of a positive, non-empty insertion set.

Let $m$ be the number of insertion sets formed by the monomer types inserted into these $\Omega(l)$ sites.
Let $c_1, c_2, \dots, c_m$ be the sums of the concentrations of the monomer types in these sets,
and $x_1, x_2, \dots, x_m$ be the number of times a monomer from each set is inserted in the subsequence. 
Then the total expected time for all of the insertions in the subsequence is $\sum_{i=1}^{m}x_i/c_i$.
Moreover, these variables are subject to the following constraints:

\begin{enumerate} \itemsep5pt
\item $\sum_{i=1}^{m}x_i = \Omega(l)$ (total number of insertions is $\Omega(l)$). 
\item $\sum_{i=1}^{m}c_i \leq 1$ (total concentration is at most 1).
\end{enumerate}

\textbf{Minimizing expected time.}
We now consider minimizing the total expected time subject to these constraints, starting with proving that $x_i/c_i = x_j/c_j$ for all $1 \leq i, j \leq m$.
That is, that the ratio of the number of sites that accept an insertion set to the total concentrations of the monomer types in the set is equal for all sets.
Assume, without loss of generality, that $x_i/c_i > x_j/c_j$ and $c_i, c_j > 0$.
Then it can be shown algebraically that the following two statements hold:

\begin{enumerate}
\item If $c_j \geq c_i$, then for sufficiently small $\varepsilon > 0$, $\frac{x_i}{c_i} + \frac{x_j}{c_j} > \frac{x_i}{c_i + \varepsilon} + \frac{x_j}{c_j - \varepsilon}$.
\item If $c_j < c_i$, then for sufficiently small $\varepsilon > 0$, $\frac{x_i}{c_i} + \frac{x_j}{c_j} > \frac{x_i}{c_i - \varepsilon} + \frac{x_j}{c_j + \varepsilon}$.
\end{enumerate} 

Since the ratios of every pair of monomer types are equal, 

$$\frac{c_i}{1} \leq \frac{c_i}{\sum_{i=1}^{m}{c_i}} = \frac{x_i}{\sum_{i=1}^{m}{x_i}} = O(x_i/l)$$

So $\Omega(l) = x_i/c_i$ and $\Omega(ml) = \sum_{i=1}^{m}x_i/c_i$.
\qed
\end{proof}

We now combine these two lemmas to prove a lower bound for deterministic systems matching the construction of Theorem~\ref{thm:types-extreme-ub} and a trade-off lower bound for systems constructing a finite set of polymers used to prove that the constructions of Theorems~\ref{thm:types-extreme-ub} and~\ref{thm:speed-extreme-ub} are optimal at both ends of the length-speed trade-off curve. 

\begin{theorem}
\label{thm:deterministic-lb}
Any polymer of length $n$ constructed by a deterministic insertion system with $k$ monomer types takes $\Omega(\log^{5/3}(n))$ expected time.
\end{theorem}

\begin{proof}
Suppose a system $\mathcal{S} = (\Sigma, \Delta, Q, R)$ has an insertion set with monomer types $m_1$ and $m_2$, and inserts $m_1$ into some polymer.
Then all polymers constructed by systems $(\Sigma, \Delta - \{m_1\}, Q, R)$ and $(\Sigma, \Delta - \{m_2\}, Q, R)$ are constructed by $\mathcal{S}$ and both have polymers not constructed by the other containing $m_2$ and $m_1$, respectively.
So any deterministic system with no unused monomer types has exclusively singleton insertion sets, i.e.\ $m = k$.
By Lemma~\ref{lem:usable-sites-ub}, $n = 2^{O(k\sqrt{k})}$ and $m = k = \Omega(\log^{2/3}(n))$.
So by Lemma~\ref{lem:speed-lb}, the expected time to construct a polymer of length $n$ is $\Omega(ml) = \Omega(\log^{5/3}(n))$.
\end{proof}

\begin{lemma}
\label{lem:trade-off-lb}
Any polymer of length $n$ constructed by an insertion system with $k$ monomer types constructing a finite set of polymers takes $\Omega(\log^2(n)/\sqrt{k})$ expected time.
\end{lemma}

\begin{proof}
By Lemma~\ref{lem:usable-sites-ub}, $n = 2^{O(m\sqrt{k})}$.
So $m = \Omega(\log{n}/\sqrt{k})$.
Constructing any polymer of length $n$ requires an insertion system of length $l = \Omega(\log{n})$. 
Then by Lemma~\ref{lem:speed-lb}, the expected time to construct any polymer of length $n$ is $\Omega(ml) = \Omega(\log^2(n)/\sqrt{k})$. 
\end{proof}

Now Lemma~\ref{lem:trade-off-lb} is combined with additional bounds on the minimum values of $k$ and $m$ to prove that the constructions in Section~\ref{sec:positive-results} are optimal in both monomer types and expected construction time for systems that construct a finite set of polymers and a polymer using the fewest monomer types (Theorem~\ref{thm:types-extreme-lb}) and the least expected time (Theorem~\ref{thm:speed-extreme-lb}).

\begin{theorem}
\label{thm:types-extreme-lb}
Any polymer constructed by an insertion system with $k$ monomer types constructing a finite set of polymers has length $2^{O(k^{3/2})}$.
Moreover, constructing a polymer of length $n = 2^{\Theta(k^{3/2})}$ takes $\Omega(\log^{5/3}(n))$ expected time.
\end{theorem}

\begin{proof}
First, observe that constructing any polymer of length $n$ in a system constructing a finite set of polymers involves an insertion sequence of length at least $\log_2(n)$ with no repeated insertion sites.
By Lemma~\ref{lem:usable-sites-ub}, $\log_2(n) = O(m\sqrt{k})$, where $m$ and $k$ are the number of monomer types used in the insertion sequence.
Since the number of insertion sets $m$ is at most $k$ (if each monomer type forms a singleton set), $m \leq k$ and $\log_2(n) = O(k\sqrt{k}) = O(k^{3/2})$.
So $n = 2^{\Theta(k^{3/2})}$ and $k = \Theta(\log^{2/3}(n))$.
By Lemma~\ref{lem:trade-off-lb}, such a polymer requires $\Omega(\log^2(n)/\sqrt{\log^{2/3}(n)}) = \Omega(\log^{5/3}(n))$ expected construction time. 
\qed
\end{proof}

Before proing an expected-time lower bound for all systems constructing finite polymer sets, we prove a helpful lemma showing that the number of insertion sets cannot be too much smaller than the number of monomer types:

\begin{lemma}
\label{lem:m-k-helper-bound}
Any insertion sequence of length $l$ with no repeated insertion sites using $k$ monomer types forming $m$ insertion sets has $m = \Omega(\sqrt{k})$.
\end{lemma}

\begin{proof}
Notice that this bound can only be obtained by assuming the monomer types are used to carry out an insertion sequence, since it is possible to have an arbitrarily large set of monomer types belonging to a single insertion set.
The number of monomer types used is at most the length of the insertion sequence ($k \leq l$), and the remainder of the proof is spent proving that the number of insertion sites in a system with $m$ insertion sets is $O(m^2)$ ($l = O(m^2)$), giving the desired inequality.

Let $\mathcal{S} = (\Sigma, \Delta, Q, R)$ be the insertion system containing the sequence.
Relabel the symbols in $\Sigma \cup \{s^* : s \in \Sigma\}$ as $s_1, s_2, \dots, s_{4k}$, with some of these symbols possibly unused.
By Lemma~\ref{lem:positive-only-sites}, $\Omega(l)$ sites are \emph{positive}: they have the form $(s_a, s_b)(s_c, \overline{s_a})$ with $b \neq c$.
 
Since the second monomer inserted to create the site must be negative, each positive site consists of at least one negative monomer type.
Let $L_i^-$ and $R_i^-$ be the sets of monomer types of the forms $(\underline{~~}, \underline{~~}, s_i, \underline{~~})^-$ and $(\underline{~~}, \overline{s_i}, \underline{~~}, \underline{~~})^-$, respectively, used in the insertion sequence of length $l$.
For a specific $i$, there exists a site of the form $(s_i, s_b)(s_c, \overline{s_i})$ only if $|L_i^-| + |R_i^-| > 0$.
So the number of values of $i$ such that a site of the form $(s_i, s_b)(s_c, \overline{s_i})$ exists is at most $\sum_{i=1}^{4k}|L_i^-| + \sum_{i=1}^{4k}|R_i^-|$.
Since all monomer types of a negative insertion set belong to the same $L_i^-$ and $R_i^-$, $\sum_{i=1}^{4k}|L_i^-| + \sum_{i=1}^{4k}|R_i^-| \leq 2m$.

Next, observe there are at most $m$ sites of the form $(s_i, s_b)(s_c, \overline{s_i})$ that accept a monomer, since each site requires a monomer from a different positive insertion set.
So the total number of positive sites that accept a monomer is at most $2m \cdot m = 2m^2$.
Since there are $\Omega(l)$ positive sites in the insertion sequence, $\Omega(l) = 2m^2$ and $l = O(m^2)$.
\qed
\end{proof}

\begin{theorem}
\label{thm:speed-extreme-lb}
Any polymer of length $n$ constructed by an insertion system constructing a finite set of polymers takes $\Omega(\log^{3/2}(n))$ expected construction time.
Moreover, constructing a polymer of length $n$ in $\Theta(\log^{3/2}(n))$ expected time requires using $\Omega(\log{n})$ monomer types.
\end{theorem}

\begin{proof}
First, observe that constructing a polymer of length $n$ in a system constructing a finite set of polymers involves an insertion sequence of length $\log_2(n) \leq l$ with no repeated sites.
By Lemmas~\ref{lem:usable-sites-ub} and~\ref{lem:m-k-helper-bound}, $\log_2(n) = O(m^2)$ and so $m = \Omega(\sqrt{\log{n}})$
Then by Lemma~\ref{lem:speed-lb}, carrying out the insertion sequence and completing the construction of the polymer takes $\Omega(ml) = \Omega(\log^{3/2}(n))$ expected time and by Lemma~\ref{lem:trade-off-lb}, $k = \Omega(\log{n})$.
\qed
\end{proof}

\section{Open Problems}
\label{sec:open-problems}

The results of Sections~\ref{sec:positive-results} and~\ref{sec:negative-results} describe the landscape of efficient polymer construction using insertion systems.
Trivial systems of just a few polymers can construct polymers of arbitrary length in optimal time, but with the caveat that the growth is uncontrolled and the systems construct infinite set of polymers. 
On the other hand, deterministically constructing a polymer of length $n$ requires $\Omega(\log^{2/3}(n))$ monomer types and $\Omega(\log^{5/3}(n))$ expected time, and both of these are achievable simultaneously. 
The intermediate situation of constructing finite sets of polymers is more intricate -- polymers can be constructed faster, but with the trade-off of using more monomer types \emph{and} non-determinism.

In our system achieving $O(\log^{3/2}(n))$ expected construction time (Theorem~\ref{thm:speed-extreme-ub}), an exponential number ($2^{\Theta(n\log\log{n})}$) of ``junk'' terminal polymers are constructed.
Since achieving such speed requires significantly fewer insertion sets than monomer types, some junk is necessary -- but how much?
One approach to proving a lower bound is to prove that insertion sites accepting large insertion sets imply a large number of terminal polymers.
We have been unable to prove such an implication even in the simplest case:

\begin{conjecture}
Every deterministic system with no unused monomer types has exclusively singleton insertion sets.
\end{conjecture}

Since assembling a polymer in $o(\log^{5/3}(n))$ expected time requires $\Omega(\log{n})$ insertions along most insertion sequences are non-deterministic, the previous conjecture implies that any improvement in speed comes with an exponential number of junk terminal polymers:

\begin{conjecture}
Any system constructing a polymer of length $n$ in $O(\log^{3/2}(n))$ expected time constructs a set of $2^{\Omega(n)}$ polymers.
\end{conjecture}

Setting aside non-determinism, the trade-off between monomer types and construction time has a lower bound (Lemma~\ref{lem:trade-off-lb}) with matching upper bounds only at the extremes.
Does there exist a parameterized system matching the lower bound across the entire range?
 
\begin{conjecture}
For every combination of $n$ and $k$ such that $\log_2^{2/3}(n) \leq k \leq \log_2(n)$, there exists a system with $k$ monomer types that constructs a polymer of length $n$ in $O(\log^2{n}/\sqrt{k})$ time. 
\end{conjecture}

In Section~\ref{sec:expressive-power}, we were able to prove that insertion systems are expressively equivalent to context-free grammars.
What minimal changes to the definition of insertion systems are necessary to achieve a Turing-universal model? 
For instance, what if we allow monomers to not only insert between existing monomers in a polymer, but replace or ``kick out'' an existing monomer or sequence of $l$ monomers?
Then we estimate a Chomsky-like hierarchy of expressive power for these ``$l$-insertion systems'':

\begin{conjecture}
The expressive power of $1$-insertion systems is equal to context-sensitive grammars.
\end{conjecture}

\begin{conjecture}
The expressive power of $l$-insertion systems for $l \geq 2$ is equal to Turing machines.
\end{conjecture}

\bibliographystyle{plain}
\bibliography{insertion_primitive}

\end{document}